%% file: main.tex
\pdfoutput=1
\documentclass[journal]{IEEEtran} 

\IEEEoverridecommandlockouts       

\usepackage{amsthm}
\usepackage{amsmath,amssymb}
\usepackage{graphicx}
\usepackage{latexsym,url}
\usepackage{sty/algorithm}
\usepackage[noend]{sty/algorithmic}
\usepackage{graphicx}
\usepackage{cite}
\usepackage{comment}
\usepackage{stfloats}
\usepackage{color}
\usepackage{xcolor}
\usepackage[hidelinks]{hyperref}
\usepackage[caption=false,font=footnotesize]{subfig}
\usepackage{bm}

\newtheorem{theorem}{Theorem}
\newtheorem{assumption}{Assumption}
\newtheorem{lemma}{Lemma}

\newtheorem{proposition}{Proposition}
\newtheorem{definition}{Definition}

\newcommand{\mc}{\mathcal}

\newcommand{\R}{{\mathbb R}}

\newcommand{\diag}{{\rm diag}}

\newcommand{\lam}{\lambda}

\newcommand{\argmin}{\mathop{\rm arg~min}\limits}

\newcommand{\red}[1]{\textcolor{red}{#1}}

\begin{document}
\title{Collision Avoidance for Ellipsoidal Rigid Bodies\\ with Control Barrier Functions Designed from\\Rotating Supporting Hyperplanes}

\author{Riku Funada$^{1}$, Koju Nishimoto$^{1}$, Tatsuya Ibuki$^{2}$, and Mitsuji Sampei$^{1}$
\thanks{*This work is supported by JSPS KAKENHI Grant Number 22K14275. (Corresponding Author: Riku Funada)}
\thanks{$^{1}$R. Funada, K. Nishimoto, and M. Sampei are with the Department of Systems and Control Engineering, Tokyo Institute of Technology, Tokyo 152-8550, Japan
{\tt\footnotesize \{\href{mailto:funada@sc.e.titech.ac.jp}{funada},\href{mailto:sampei@sc.e.titech.ac.jp}{sampei}\}@sc.e.titech.ac.jp},
\href{mailto:nishimoto@sl.sc.e.titech.ac.jp} {\tt\footnotesize {nishimoto}@sl.sc.e.titech.ac.jp}}
\thanks{$^{2}$T. Ibuki is with the Department of Electronics and Bioinformatics, Meiji University, Kanagawa 214-8571, Japan
\href{mailto:ibuki@meiji.ac.jp} {\tt\footnotesize {ibuki}@meiji.ac.jp}}
}
\maketitle
\IEEEpeerreviewmaketitle

\input{text/0-Abstract}
\input{text/1-Introduction}

\input{text/2-Preliminary}
\input{text/3-Problem_Formulation}

\input{text/4-Proposed_CBF}

\input{text/5-Simulation}

\input{text/6-Conclusion}
\appendices
\input{text/Z_Appendix}

\input{text/Z-Appendix_A}

\bibliographystyle{sty/IEEEtran}
\bibliography{biblio}

\end{document}

%% file: text/0-Abstract.tex
\begin{abstract}
This paper proposes a collision avoidance method for ellipsoidal rigid bodies, which utilizes a control barrier function (CBF) designed from a supporting hyperplane. We formulate the problem in the Special Euclidean Group $SE(2)$ and $SE(3)$, where the dynamics are described as rigid body motion (RBM). Then, we consider the condition for separating two ellipsoidal rigid bodies by employing a signed distance from a supporting hyperplane of a rigid body to the other rigid body. Although the positive value of this signed distance implies that two rigid bodies are collision-free, a naively prepared supporting hyperplane yields a smaller value than the actual distance. To avoid such a conservative evaluation, the supporting hyperplane is rotated so that the signed distance from the supporting hyperplane to the other rigid body is maximized. We prove that the maximum value of this optimization problem is equal to the actual distance between two ellipsoidal rigid bodies, hence eliminating excessive conservativeness.
We leverage this signed distance as a CBF to prevent collision while the supporting hyperplane is rotated via a gradient-based input. 
The designed CBF is integrated into a quadratic programming (QP) problem, where each rigid body calculates its collision-free input in a distributed manner, given communication among rigid bodies.
The proposed method is demonstrated with simulations. Finally, we exemplify our method can be extended to a vehicle having nonholonomic dynamics.
\end{abstract}


%% file: text/1-Introduction.tex
\section{Introduction} \label{sec:intro}

Collision avoidance is one of the fundamental requirements for ensuring the safe operation of multi-robot systems in many application fields, including precision agriculture \cite{Zhang12}, autonomous transportation \cite{Miyano20, Gong17}, and environmental monitoring \cite{Funada20}. In such challenging and complex domains, it is paramount important to integrate robots having different capabilities, sizes, and shapes into a system to complete the task \cite{Rizk19}. In order to embrace such heterogeneity, collision avoidance methods are required to guide robots not to collide with each other while considering their shapes.

Real-time collision avoidance has been actively explored in the fields of multi-robot systems, but most work approximates the shape of a robot as a circle or sphere, as will be discussed in Section~\ref{ssec:Related}.
Although such approaches can be utilized for any robots by overestimating the original form of robots to a sphere enclosing them, this approach could result in too conservative evasion if they have a nonspherical, especially elongated, body.
To model such heterogeneous shapes with a higher fidelity model and still a small set of parameters, an ellipsoidal approximation of the shape has been utilized in SLAM and path planning fields.

The work \cite{Preiss17} proposes the path planning method, where each quadcopter is modeled as an ellipsoid to account for a downwash wind they generate. 
The approach to model a rigid body as an ellipsoid is also suitable for many SLAM or environmental monitoring techniques, e.g., \cite{Hatanaka2016,Nicholson2019,Ok2019}, where these works represent the objects 
as an ellipsoid and infer their sizes from visual information.






Despite the advantages of modeling the shape of rigid bodies as ellipsoids, the difficulties in deriving the distance between two separated ellipsoids hinder the development of collision avoidance controllers for ellipsoidal rigid bodies. 
In general, the analytical solutions of the distance between two ellipses/ellipsoids are difficult to obtain in simple analytical form, and the algorithms calculating the numerical solutions of the distance are utilized in the computational science field instead \cite{Choi2020, GIRAULT22}. 
For two-dimensional ellipses, the method in \cite{Choi2005} detects whether two ellipses overlap by evaluating the discriminant of a cubic characteristic equation. 
However, because the value of the discriminant is not necessarily in a proportional relationship with the distance, rigid bodies could take unreasonable motion if we employ it in a collision avoidance method, as will be shown in the comparative studies with the work~\cite{Ibuki22} later. 
In addition, it is difficult to extend this method to ellipsoids in 3D environments because one is required to solve the discriminant of higher dimensional characteristic equations, which solution is divided by case. 






This paper presents a collision avoidance method for elliptical/ellipsoidal rigid bodies in 2D/3D environments shown in Fig.~\ref{fig:scenario}, the dynamics of which are described by rigid body motion (RBM). To circumvent the difficulty of directly deriving analytical form of the distance between two ellipsoids, this paper proposes a novel CBF that utilizes a signed distance from a supporting hyperplane of an ellipsoid to the other ellipsoid, as depicted in Fig.~\ref{fig:ellipse_rot}. Because a naively prepared supporting hyperplane could render a smaller distance than the actual distance between two ellipsoidal rigid bodies (Fig.~\ref{fig:ellipse_rot}\subref{subfig:ellipse_rot_1}), we design a gradient-ascent-based update law, where the supporting hyperplane is rotated so that the signed distance from the supporting hyperplane to the other ellipsoidal rigid body is maximized (Fig.~\ref{fig:ellipse_rot}\subref{subfig:ellipse_rot_2}). We prove that the maximum signed distance derived from this optimization problem is equal to the actual distance between two ellipsoidal rigid bodies. A novel CBF integrating the rotating supporting hyperplane is incorporated into the quadratic programming (QP) problem, which allows a distributed computation under communication. The simulation studies demonstrate that the proposed method can successfully achieve collision avoidance for elliptical/ellipsoidal rigid bodies in 2D/3D environments. Finally, we exemplify the proposed method can be extended to a vehicle having nonholonomic dynamics.


\begin{figure} 
    \centering
    \includegraphics[width=50mm]{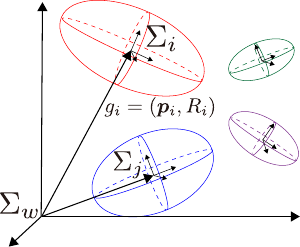}
    \caption{Proposed scenario. The rigid bodies characterized as ellipses or ellipsoids with heterogeneous shapes avoid collisions with each other.}
    \label{fig:scenario}
\end{figure}

\subsection{Related Work} \label{ssec:Related}

Among the various methodologies for preventing collisions, artificial potential fields (APFs) are one of the most traditional methods, which were first presented in \cite{Khatib1986}. 
APFs have been utilized for multi-robot (multi-agent) systems to attain cooperative behaviors \cite{Hoy2015}, such as flocking \cite{Olfari-saber06, Gazi04} and formation control \cite{Dimarogonas08}.
As a repulsive potential field is a key element for designing collision avoidance behaviors, several functions are proposed, including \cite{Stipanovic07}, which is activated only when any other robots is in the sensing region of the robot.
The work \cite{Arslan19} designs a vector field for navigating a circular robot in an environment cluttered with convex obstacles by utilizing the idea of Voronoi diagrams. 
Still, such a Voronoi-based approach cannot be readily extendable for the robots with ellipsoidal shapes because it is known to be challenging to generate the Voronoi diagram for ellipsoids~\cite{Ioannis07}. 

Control barrier functions (CBFs) are another popular method to guarantee the safety of the system by formulating a quadratic programming (QP) problem \cite{Ames2017,Ames2019_CBF_thapp,Egerstedt21}. 
Extensive studies utilize a CBF-based framework to achieve collision avoidance among multi-robot systems. 
For example, the work \cite{Wang2017} achieves a distributed collision-free coordination in multi-robot systems with heterogeneity in the control input range. 
The robots with limited sensing ranges are also considered in \cite{Glotfelter19} with employing the hybrid CBFs.

The work \cite{Singletary2020} provides the comparative analysis between CBFs and APFs, which proves that one can obtain CBFs from a given APF.
The authors also claim that CBFs designed from APFs have additional beneficial properties, such as mitigation of oscillation, 
compared with its counterpart of APFs.
In addition, CBFs have the capacity to embrace different types of safety requirements, including collision avoidance, prevention of battery depletion \cite{Notomista21}, and connectivity maintenance \cite{Ong21}, as one synthesized controller.
Because of these virtues, this paper opts for the CBF-based approach. 
Still, most of the papers mentioned above model the robot as a point or a sphere. 



The work \cite{Do2013} addresses the flocking for ellipsoidal rigid bodies, where the APF-based method is employed for collision avoidance. However, the condition utilized for designing the repulsive potential function takes complicated form and is not straightforwardly extendable to CBFs.
The authors in \cite{Verginis2019} develop a collision avoidance methodology by utilizing the result in the computer graphics fields, which provides the separating conditions between ellipsoids. Still, the metric used in the separating condition does not intuitively align with the distance, and its physical interpretation might be difficult to understand. 
The work \cite{Dhal21} considers planar robots having quadratic surfaces, where a condition calculated from robots' relative velocities, collision cone, is utilized for avoiding a collision. However, the condition requires several procedures to derive and is not easy to employ as CBFs.

\begin{figure}[t!]
    \vspace{-0.25cm}
    \centering
    \subfloat[Initial supporting hyperplane]
    {\makebox[0.47\hsize][c]{\includegraphics[width=0.45\linewidth]{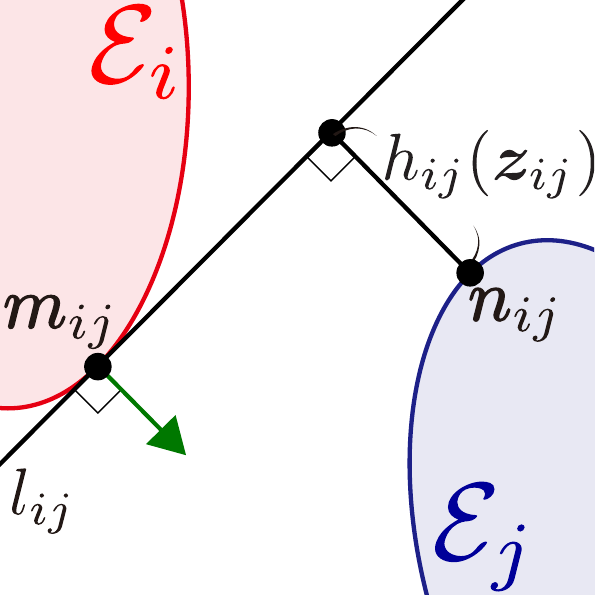}
    \label{subfig:ellipse_rot_1}}} 
    ~
    \subfloat[Updated supporting hyperplane]
    {\makebox[0.47\hsize][c]{\includegraphics[width=0.45\linewidth]{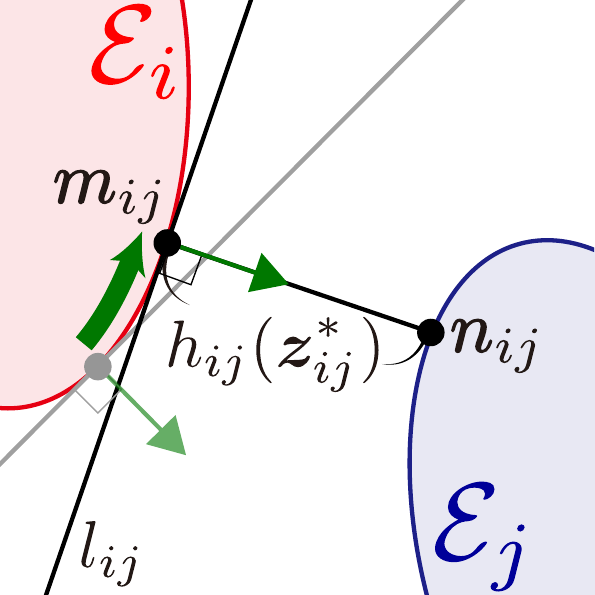}
    \label{subfig:ellipse_rot_2}}}
    \caption{The supporting hyperplane $l_{ij}$ separating two elliptical rigid bodies $\mc E_i$ and $\mc E_j$. 
    Both (a) and (b) show the distance $h_{ij}(\bm{z}_{ij})$ between the ellipse $\mc E_j$ and a supporting hyperplane $l_{ij}$. The tangent point is denoted as $\bm{m}_{ij}$, which is specified by a unit vector $\bm{z}_{ij}$ as detailed in Section~\ref{sec:ellip}.}
    \label{fig:ellipse_rot}
\end{figure}

The work \cite{Srinivasan21} presents the extent-compatible CBF, which can enforce the safety of the robot having volume. 
The proposed method relies on the sum-of-squares (SOS) based optimization method and could be applied to ellipsoidal rigid bodies. Still, the computational burden in the SOS problem might hinder the application to a team composed of many robots. 
The work~\cite{Thirugnanam2022}, which coincidentally was presented around the same time as our preliminary work~\cite{Nishimoto22}, considers collision avoidance for polygonal robots, where a nonsmooth CBF is utilized. 
The authors in~\cite{Thirugnanam2022} developed this approach to handle general convex robots in \cite{Thirugnanam2023}, where the input is calculated in a centralized manner.
%
The paper~\cite{Dai2023safe} also considers collision avoidance for a general convex shape by employing a scaling factor instead of the distance.  
Both \cite{Thirugnanam2023} and \cite{Dai2023safe} require solving an additional optimization problem other than a QP yielding a collision-free input, hence necessitating an extra computational effort.

Shifting the focus to spacecraft navigation fields, several studies consider collision avoidance between a spacecraft, modeled as a point or a circle, and debris represented as circles or ellipses~\cite{Park11,Park2016}. The work~\cite{Park2016} sets a rotating supporting hyperplane 
on the ellipsoidal debris to prevent collisions. Since the supporting hyperplane linearizes the constraints for collision avoidance, this formulation reduces the computation time of a navigation controller. 
Still, typical methodologies for preparing supporting hyperplanes require a user-specified constant angular velocity for rotating hyperplanes or a predefined spacecraft trajectory for setting multiple hyperplanes along a path~\cite{Park2016, Zagaris18,Malyuta2021}. 
In addition, the spacecraft is modeled as a point or sphere; hence this approach cannot be readily extendable for collision avoidance for ellipsoidal rigid bodies.

This paper presents a novel CBF that achieves collision avoidance for ellipsoidal rigid bodies while leveraging the simplicities of the rotating supporting hyperplane method. The update rule for the supporting hyperplane is newly developed by providing theoretical results eliminating conservativeness and guaranteeing safety. 
The proposed method only requires solving a QP to obtain a collision-free input; hence its computational effort is low, as same as traditional CBFs modeling a robot as a circle. 
Furthermore, we formulate the problem in the Special Euclidean Group to represent a collision avoidance law for the rigid bodies' poses in 2D and 3D environments. 
More detailed explanations of the contributions are shown below.




\subsection{Contributions} \label{ssec:contri}
This paper develops our preliminary work~\cite{Nishimoto22}, which only considered collision avoidance for elliptical robots in a 2D environment, where the dynamics of the robots were modeled as a single integrator. In addition, the presented collision avoidance law only allowed centralized computation. The contributions of this paper are as follows.
\begin{itemize}
    \item We formulate the collision avoidance problem in the Special Euclidean Group, where the dynamics of rigid bodies are modeled as rigid body motion (RBM) to design a unified collision avoidance law for the robot's pose (position and attitude) in 2D and 3D environments. With this extension, we newly derive the time derivative of the CBF along with the dynamics represented as RBM, which is presented in Lemma~\ref{lem:dot_h}. 
    \item A novel CBF utilizing a rotating supporting hyperplane is developed, which modifies the one in our conference work~\cite{Nishimoto22} to account for a 3D case. We prove that the maximum of the optimization problem considered in the update rule of the supporting hyperplane is equal to the actual distance between two rigid bodies and ensure the forward invariance of the set representing no collisions.
    \item We show that the QP yielding a collision-free input can be calculated in a distributed manner by assuming the communication between rigid bodies. Furthermore, we prove the validity of the proposed CBF, namely the feasibility of the presented QP in both 2D and 3D scenarios.
    \item More comprehensive simulation studies are presented to demonstrate the effectiveness of the proposed method. 
    \item We extend the proposed method to achieve collision avoidance for nonholonomic systems by employing a vehicle as a case study, where the update rule of the supporting hyperplane is modified.
\end{itemize}





%% file: text/2-Preliminary.tex
\section{Preliminary: Control Barrier Function} \label{sec:Preliminary}

This section introduces a CBF, which will be utilized to guarantee the collision avoidance of the robots. 
Together with a basic CBF, we explain an approach to guarantee safety described by the high-relative degree constraints, which will be utilized when we extend our method to a vehicle having the nonholonomic dynamics in Section~\ref{sec:vehicle}.
Note that this paper considers rigid bodies in the Special Euclidean Group, whose dynamics are expressed by a different form with \eqref{eq:cont_affine}. Still, the same approach can be utilized by calculating the time derivatives of CBFs along the trajectories of the system described by RBM, as will be shown in Lemma~\ref{lem:dot_h}. 

Let us consider the control affine system
\begin{align} \label{eq:cont_affine}
    \dot{\bm{x}} = f(\bm{x}) + g(\bm{x})\bm{u},
\end{align}
where $f$ and $g$ are locally Lipschitz, $\bm{x} \in \R^n$ and $\bm{u} \in \R^m$. 
We also introduce a set defined as the zero super-level set of a continuously differentiable function $h(\bm{x})$, namely, $\mc S = \{ \bm{x} \in \R^n \mid h(\bm{x}) \geq 0 \}$.
Then, a CBF is defined as follows. 
\begin{definition} \cite[Def. 5]{Ames2017} \label{def:CBF}
The function $h$ is a control barrier function (CBF) defined on a set $\bar{\mathcal{S}}$ with $\mathcal{S} \subseteq \bar{\mathcal{S}} \subset \mathbb{R}^n$, if there exists an extended class $\mathcal{K}$ function $\alpha$, such that for the control system~\eqref{eq:cont_affine}
\begin{align}
    \begin{split}
    &\sup_{\bm{u}}\dot h(\bm{x}, \bm{u})  = \sup_{\bm{u}}\left[L_fh(\bm{x})+L_gh(\bm{x})\bm{u}\right]\geq -\alpha(h(\bm{x})),\\
    &\hspace{6.4cm} \forall \bm{x}\in\bar{\mathcal{S}},  \label{CBFdefine}
    \end{split}
\end{align}
where $L_f h(\bm{x})$ and $L_g h(\bm{x})$ are the Lie derivatives of $h$ along $f(\bm{x})$ and $g(\bm{x})$, respectively.
\end{definition}

The forward invariance of the set $\mc S$, defined just below \eqref{eq:cont_affine}, can be achieved through the following proposition.
\begin{proposition}\cite[Cor. 2]{Ames2017} 
Given a set $\mc S$, if $h$ is a CBF on $\bar{\mc S}$, then any Lipschitz continuous controller $u(\bm{x}): \bar{\mc S} \to \R^m$ such that 
\begin{align} \label{eq:CBF_cond}
\dot h(\bm{x}, \bm{u}) = L_f h(\bm{x}) + L_g h(\bm{x})\bm{u}(\bm{x})  \geq -\alpha( h(\bm{x}) )
\end{align}
will render the set $\mc S$ forward invariant.
\end{proposition}


The condition \eqref{eq:CBF_cond} can be integrated into the control law to ensure the forward invariance of the set by leveraging Quadratic Programming (QP). Let us denote the nominal input as $\bm{u}_{\rm nom}$, and wish to modify it minimally invasive way so that the condition \eqref{eq:CBF_cond} is satisfied.
This objective can be achieved by employing the input $\bm{u}^*$ obtained from the following QP.
\begin{subequations} \label{eq:QP_default}
\begin{align}
    \bm{u}^* =& \argmin_{\bm{u}}~\|\bm{u}-\bm{u}_{\mathrm{nom}}(\bm{x})\|^2\\
    &~\mbox{s.t.}~L_fh(\bm{x})+L_gh(\bm{x})\bm{u}\geq -\alpha(h(\bm{x})) \label{eq:QP_general_const}
\end{align}
\end{subequations}

The previous discussion describes how CBFs can be utilized to ensure the safety constraints for a control affine system~\eqref{eq:cont_affine}. 
However, so far, we presume that the control input appears in the first derivative of the CBF with respect to time, as in \eqref{eq:CBF_cond}. 
This property can be formally expressed as having a relative degree one, which is defined as follows.
\begin{definition} \cite[Def.~5.2]{Egerstedt21}
The system $\dot{\bm{x}} = f(\bm{x}) + g(\bm{x})\bm{u}$, with output $y=h(\bm{x})$, has relative degree, $r \in {\mathbb Z_+}$, at $\bm{x}_0$, if 
\begin{align}
    L_g L_f^\delta h(\bm{x}) = 0,~\forall \delta \leq r-2,
\end{align}
$\forall \bm{x}$ in a neighborhood of $\bm{x}_0$, and 
\begin{align}
    L_g L_f^{r-1} h(\bm{x}_0) \neq 0.
\end{align}
\end{definition}
The QP shown in \eqref{eq:QP_default} cannot be employed if the safety constraint is not of relative degree one, because the control input does not appear in the first derivative of CBFs~\cite{Nguyen2016,Xiao2022,Notomista21}. 
To overcome this issue, the framework of a basic CBF needs to be augmented.
The rest of this section introduces the approach presented in~\cite{Notomista21} to grant safety for a higher relative degree system, which will be utilized in Section~\ref{sec:vehicle}.


Let us first consider when the relative degree of the system is two. To guarantee the forward invariance of the safe set $\mc S_1 = \{ \bm{x} \in \R^n \mid h_1(\bm{x}) \geq 0 \}$, the first derivative of the CBF $h_1(\bm{x})$ should satisfy the following condition.
\begin{align} \label{eq:CBF_issue_degree2}
    \dot h_1(\bm{x}) + \alpha_1(h_1(\bm{x})) = L_f h_1(\bm{x}) + \alpha_1(h_1(\bm{x})) \geq 0
\end{align}
To ensure the condition \eqref{eq:CBF_issue_degree2}, let us define an additional CBF, where we set $\alpha_1(h_1(\bm{x})) = \gamma_1 h_1(\bm{x})$ for simplicity, as
\begin{subequations}
\begin{align}
    h_2(\bm{x}) &= \dot h_1(\bm{x}) + \gamma_1 h_1(\bm{x}) \\
    &=L_f h_1(\bm{x}) + \gamma_1 h_1(\bm{x}),
\end{align}
\end{subequations}
whose zero superlevel set is $\mc S_2 = \{ \bm{x} \in \R^n \mid h_2(\bm{x}) \geq 0 \}$.
If there exists a positive constant $\gamma_1$ and a locally Lipschitz extended class $\mc K$ function $\alpha_2$ such that
\begin{align} \label{eq:CBF_cond_degree2}
\begin{split} 
    &\sup_{\bm{u}} \left[ \dot h_2(\bm{x}, \bm{u}) + \alpha_2(h_2(\bm{x})) \right] \\
    &=\sup_{\bm{u}} \left[ L_f^2 h_1(\bm{x}) + L_g L_f h_1(\bm{x}) \bm{u} \right.\\
    &\hspace{3cm} \left. + \gamma_1 L_f h_1(\bm{x}) + \alpha_2(h_2(\bm{x})) \right] \geq 0,
\end{split}
\end{align}
then $h_2$ is a valid CBF. 
This implies that by employing 
\begin{align} 
    \dot h_2(\bm{x}, \bm{u}) + \alpha_2(h_2(\bm{x})) \geq 0
\end{align}
as the constraint of the QP, instead of \eqref{eq:QP_general_const}, the resultant input $\bm{u}^*$ renders the set $\mc S_2$ forward invariant and, in turn, the set $\mc S_1$ too. The following proposition generalizes this technique.

\begin{proposition} \cite[Thm.~1]{Notomista21}
 Given a dynamical system \eqref{eq:cont_affine}, a sufficiently smooth CBF $h_1(\bm{x})$ with relative degree $r$ and a CBF $h_r(\bm{x})$ that can be evaluated recursively starting from $h_1(\bm{x})$ using the following equation: 
\begin{align}
     h_{\delta + 1}(\bm{x}) = \dot h_{\delta}(\bm{x}) + \alpha_\delta( h_\delta(\bm{x}) ),~ 1\leq \delta < r
\end{align}
 with $\alpha_\delta$ continuously differentiable extended class $\mc K$ functions, we define the set $K_r(\bm{x})$ as 
 \begin{align}
     \begin{split}
     &K_r (\bm{x}) = \Bigg\{ \bm{u} \in \R^m ~\Bigg|~ L_f^r h_1(\bm{x}) + L_g L_f^{r-1} h_1(\bm{x}) \bm{u} \\
     &+\!\sum_{i=1}^{r-1} 
     \sum_{J \in \bigl(
\begin{smallmatrix}
   r-1 \\
   i
\end{smallmatrix}
\bigl) } \prod_{j \in J} 
     \frac{\partial \alpha_j}{\partial h_j} L_f^{r-i} h_1(\bm{x}) \!+\! \alpha_r( h_r(\bm{x}) ) \!\geq\! 0 \Bigg\}, 
     \end{split}
 \end{align}
where $\bigl(\begin{smallmatrix}
   r-1 \\
   i
\end{smallmatrix}\bigl)$
is the set of $i$ combinations from the set $\{ 1\cdots {r-1} \} \subset \mathbb{N}$ and $\alpha_r$ is a locally Lipschitz extended class $\mc K$ function. Then, any Lipschitz continuous controller $u \in K_r(\bm{x})$ will render the set $\mc S_1 = \{ \bm{x} \in \R^n \mid h_1(\bm{x}) \geq 0 \}$ forward invariant.
\end{proposition}

%% file: text/3-Problem_Formulation.tex
\section{Problem Formulation}

This paper considers a collision avoidance method for rigid bodies, the shape of which can be modeled by an ellipse or an ellipsoid. 
The rigid bodies, labeled through the index set $\mc N = \{1 \cdots n\}$, are distributed in the Euclidean space $\R^d$, as illustrated in Fig.~\ref{fig:scenario}. 
Note that this paper considers the scenario with $d = 2$ or $d = 3$. 
We denote the world coordinate frame as $\Sigma_w$. 
The coordinate frame of rigid body~$i$ is defined as $\Sigma_i$, which is arranged at the center of rigid body~$i$ so that its axes aligned with each axis of an ellipse or an ellipsoid.
The relative pose of $\Sigma_i$ with respect to $\Sigma_w$ is described as $g_i = (\bm{p}_i, R_i) \in SE(d):=\R^d \times SO(d)$ with the position $\bm{p}_i \in \R^d$ and the orientation $R_i\in SO(d):=\{ R\in \R^{d\times d} \mid RR^\top = I_d,~\det{(R)} = 1  \}$. 
Hereafter, we denote a two-dimensional ellipse and three-dimensional ellipsoid as an ellipsoid all together to make notations simpler.

The area rigid body~$i$ occupies is modeled as an ellipsoid described as
\begin{align} \label{eq:set_ellipsoid}
    \mc E_i \!=\! \left\{ \bm{q} \!\in\! \R^d \mid (\bm{q}\!-\!\bm{p}_i )^\top R_i Q_i^{-2} R_i^\top (\bm{q}\!-\!\bm{p}_i ) \!- \!1 \leq 0 \right\},
\end{align}
where $Q_i$ is a diagonal matrix having $q_{im}$ as the $m$-th diagonal element corresponding with the length of the $m$-th axis of the ellipsoid.
As detailed later, we assume that rigid body~$i$ can obtain the pose $(\bm{p}_j, R_j)$ and shape $Q_j$ of other rigid bodies $j\in \mc N\backslash \{i\}$ through sensing or communications.

Let us denote the body velocity of rigid body~$i$ relative to $\Sigma_w$ as $V_{i}^b = [\bm{v}_{i}^\top~\bm{\omega}_{i}^\top]^\top \in \R^{d+\frac{d(d-1)}{2}}$, where $\bm{v}_{i}\in \R^d$ and $\bm{\omega}_{i}\in \R^{\frac{d(d-1)}{2}}$ are the translational and angular body velocity, respectively\footnote{While an angular body velocity is a scalar for $d=2$, making a bold letter unsuitable, we use $\bm{\omega}$ for both $d=2$ and $d=3$ for notational simplicity.}.
We also introduce the operator $\wedge$, which renders a skew-symmetric matrix, i.e., an element of $so(d):= \{ S \in \R^{d\times d} \mid S + S^\top = O_d\}$. 
More specifically, $\wedge$ renders $so(2)$ from $a \in \R$ and $so(3)$ from $\bm{a} = [a_1~a_2~a_3]^\top \in \R^3$ as
\begin{align} \label{eq:wedge}
    \hat a = 
    \begin{bmatrix}
    0 & -a\\
    a & 0
    \end{bmatrix},~
    \hat{\bm{a}} = 
    \begin{bmatrix}
    0 & -a_3 & a_2 \\
    a_3 & 0 & -a_1 \\
    -a_2 & a_1 & 0
    \end{bmatrix}.
\end{align}
We also define the inverse operator of $\wedge$ as $\vee$.
By employing $\wedge$, the pose $g_i$ and the body velocity $V_{i}^b$ can be described by the following homogeneous representation.
\begin{subequations}\label{eq:RBM_homoge}
\begin{align} 
    g_i &= 
    \begin{bmatrix}
    R_i & \bm{p}_i \\
    0 & 1
    \end{bmatrix} \in \R^{(d+1)\times (d+1)} \label{eq:gi_homoge}\\
    \hat V_{i}^b &= 
    \begin{bmatrix}
    \hat{\bm{\omega}}_i & \bm{v}_i \\
    0 & 0
    \end{bmatrix} \in \R^{(d+1)\times (d+1)}
\end{align}
\end{subequations}
Then, the dynamics of rigid body~$i$ can be modeled as the rigid body motion \cite{Hatanaka2015_Springer} 
\begin{align} \label{eq:dyn_RBM}
    \dot g_{i} = g_{i} \hat V_{i}^b.
\end{align}
Except for Section~\ref{sec:vehicle} exemplifying the proposed CBF can be extendable to a nonholonomic system, the dynamics of rigid body~$i$ is governed by \eqref{eq:dyn_RBM}, where the body velocity $V_i^b$ constitutes a control input for rigid body~$i$.


This paper proposes a collision avoidance method for a team of rigid bodies described with \eqref{eq:set_ellipsoid}.
If we can obtain the minimum distance between $\mc E_i$ and $\mc E_j$ as $w_{ij}^*(g_i, g_j)$, the safe set preventing collisions between rigid bodies $i$ and $j$ can be described as
\begin{align}\label{eq:safeset}
    \mathcal{S}_{ij}=\left\{(g_i, g_j) \in SE(d) \times SE(d) \mid w_{ij}^*(g_i,g_j)\geq 0\right\},
\end{align}
where this set has to be rendered forward invariant.

%% file: text/4-Proposed_CBF.tex
\section{CBFs for Ellipsoidal Rigid Bodies} \label{sec:ellip}

\begin{figure}[t!]
    \vspace{-0.1cm}
    \centering
    \subfloat[$h_{ij}(g_i, g_j, \bm{z}_{ij})<0$]
    {\makebox[0.48\hsize][c]{\includegraphics[width=0.4\linewidth]{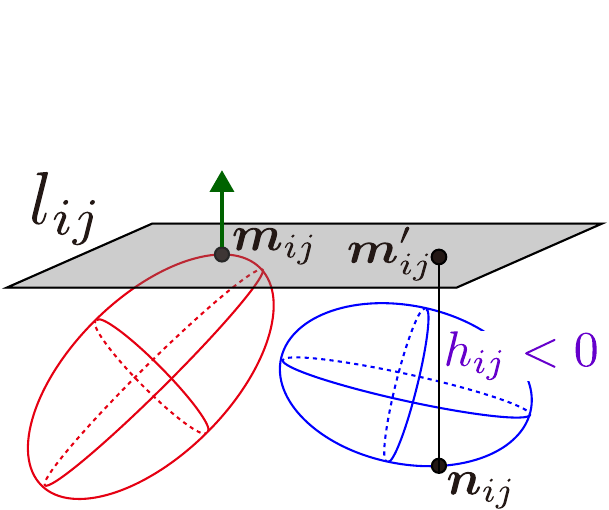}
    \label{subfig:h_exp_same}}} 
    ~
    \subfloat[$h_{ij}(g_i, g_j, \bm{z}_{ij})>0$]
    {\makebox[0.48\hsize][c]{\includegraphics[width=0.4\linewidth]{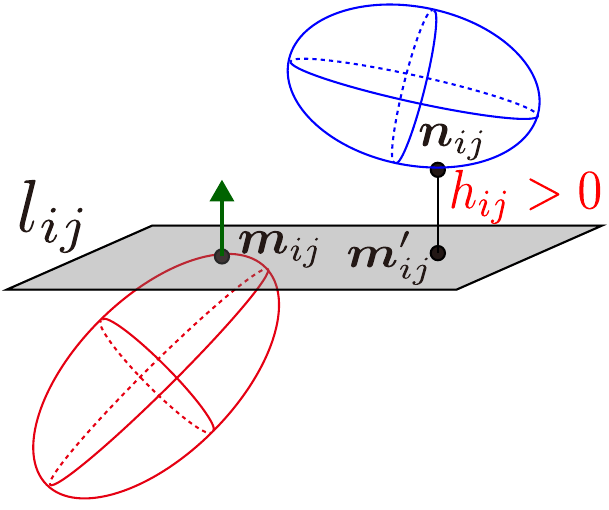}
    \label{subfig:h_exp_diff}}}
    \caption{The separation condition evaluated by $h_{ij}$, the minimum signed distance from the supporting hyperplane $l_{ij}$. 
    $h_{ij}$ returns zero on $l_{ij}$ while it takes the larger value as a point to be evaluated moves to the upper direction specified by the green normal vector at a point $\bm{m}_{ij}$.
    Note that if the supporting hyperplane $l_{ij}$ is not adequately prepared, $h_{ij}$ could become negative even when two rigid bodies are separated, as shown in (a).}
    \label{fig:h_exp}
\end{figure}

This section presents a novel CBF, which enforces the forward invariance of the set $\mc S_{ij}$, namely preventing collisions between ellipsoidal rigid bodies $i$ and $j$. 
If $w_{ij}^*$ in \eqref{eq:safeset} can be derived in analytical form, we can employ it as a CBF. 
However, as mentioned in Section~\ref{sec:intro} and \cite{Choi2020, GIRAULT22,Choi2005}, it is difficult to obtain the analytical solution of the distance between two ellipsoids (for both $d=2$ and $d=3$). 
Moreover, the numerical solutions of $w_{ij}^*$ cannot be utilized as a CBF either since the CBF-constraint in \eqref{eq:QP_default} requires to calculate the derivative of the CBF. 
To circumvent the difficulties, we formulate a novel CBF that employs a surrogate distance that can be derived in analytical form.
More specifically, we utilize the signed distance from a supporting hyperplane of an ellipsoid to the other one, depicted as $h_{ij}$ in Fig.~\ref{fig:h_exp}.
Because $h_{ij}$ could take a shorter length than $w_{ij}^*$ with a naively prepared supporting hyperplane, we propose the framework that drives $h_{ij}$ to $w_{ij}^*$ based on the gradient of a signed distance, as shown in Fig.~\ref{fig:ellipse_rot}.

\subsection{Supporting Hyperplanes Separating Two Ellipsoids}

We first define a supporting hyperplane $l_{ij}$ of an ellipsoid, which touches $\mc E_i$ at the point $\bm{m}_{ij}$, as shown in Fig.~\ref{fig:v_m_trans}.
Let us define the point $\bm{m}_{ij}$ as
\begin{align}
    \bm{m}_{ij}(g_i,\bm{z}_{ij})&=\bar{Q}_i(R_i) \bm{z}_{ij}+\bm{p}_i,  \label{eq:m_ij}\\
    \| \bm{z}_{ij}\| &= 1, \label{eq:v_ij}
\end{align}
where $\bar{Q}_i(R_i) = R_i Q_i R_i^\top$ is a positive definite matrix. 
Hereafter, for notational simplicity, we will denote $\bar{Q}_i(R_i)$ as $\bar{Q}_i$.
Note that $\bar{Q}_i^\top \bar{Q}_i = \bar{Q}_i^2 = R_i Q_i^2 R_i^\top$ holds since $Q_i$ is a diagonal matrix.
The unit vector $\bm{z}_{ij} \in \mathbb S^{(d-1)}$ specifies a point on the boundary of the ellipsoid $\mc E_i$ as shown in Fig.~\ref{fig:v_m_trans}.
Then, the supporting hyperplane $l_{ij}$ is expressed as 
\begin{align}
    l_{ij} \!=\! \left\{\bm{q}\in\mathbb{R}^d~|~\bm{z}_{ij}^\top \bar{Q}_i^{-1} \bm{q}-\left(1+\bm{z}_{ij}^\top \bar{Q}_i^{-1} \bm{p}_i\right) \!=\! 0\right\}, \label{hyperplane}
\end{align}
which is determined by $g_i$ and $\bm{z}_{ij}$.

\begin{figure} 
    \centering
    \includegraphics[width=52mm]{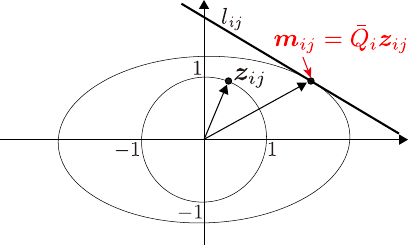}
    \caption{Relationship between a unit vector $\bm{z}_{ij}$ and a point $\bm{m}_{ij}$ on an ellipse characterized by $\bar{Q}_i$. 
    $\bm{z}_{ij}$ in \eqref{eq:v_ij} specifies a point on the unit circle. Then, $\bm{z}_{ij}$ is transformed with the positive definite matrix $\bar{Q}_i$ to designate the point on the ellipse, which is adopted as a tangent point of a supporting hyperplane $l_{ij}$. Here, we omit $\bm{p}_i$ by setting the center of the ellipse $\bm{p}_i=[0~0]^\top$.}
    \label{fig:v_m_trans}
\end{figure}

Let us utilize the supporting hyperplane $l_{ij}$ to derive a collision-free condition, in which two ellipsoids are separated by the supporting hyperplane. 
For this goal, we calculate the signed distance from the supporting hyperplane $l_{ij}$, which yields a positive value to a point in a different half-space with $\mc{E}_i$, and a negative value otherwise, as in Fig.~\ref{fig:h_exp}.
This signed distance from $l_{ij}$ to $\mc{E}_j$ renders the minimum value when the distance is evaluated with the point $\bm{n}_{ij} \in \mc E_j$ as 
\begin{align}
    \bm{n}_{ij}(g_i,g_j,\bm{z}_{ij}) = -\frac{1}{\left\|\bar{Q}_j \bar{Q}_i^{-1}\bm{z}_{ij}\right\|}\bar{Q}_j^2\bar{Q}_i^{-1}\bm{z}_{ij} + \bm{p}_j.
\end{align}
The minimum signed distance from $l_{ij}$ is then described as
\begin{align}
    h_{ij}(g_i, g_j, \bm{z}_{ij})\!=\!\frac{ -\left\|\bar{Q}_j \bar{Q}_i^{-1} \bm{z}_{ij}\right\| \!+\! (\bm{p}_j-\bm{p}_i)^\top\bar{Q}_i^{-1} \bm{z}_{ij} \!-\! 1}{\left\|\bar{Q}_i^{-1}\bm{z}_{ij}\right\|}. \label{eq:CBF_cand}
\end{align}
As shown in Fig.~\ref{fig:h_exp}, $h_{ij}(g_i, g_j, \bm{z}_{ij})$ yields a positive value if and only if $\mc{E}_i$ and $\bm{n}_{ij}$ exist in the different half-space divided by $l_{ij}$.
Note that the point $\bm{n}_{ij} \in \mc E_j$ is not the nearest point from the supporting hyperplane $l_{ij}$ (Fig.~\ref{fig:h_exp}\subref{subfig:h_exp_same}) because this signed distance increases along with the direction of the normal vector of $l_{ij}$, depicted as the green arrow in Fig.~\ref{fig:h_exp}.

Since $h_{ij}(g_i,g_j, \bm{z}_{ij}) > 0$ signifies $\mc{E}_i$ and $\mc{E}_j$ are separated by the supporting hyperplane defined with $\bm{z}_{ij}$, the function $h_{ij}(g_i,g_j, \bm{z}_{ij})$ could be employed as a CBF to achieve collision avoidance.
However, this condition could be a conservative if $l_{ij}$ is prepared naively as shown in Fig.~\ref{fig:h_exp}\subref{subfig:h_exp_same}, where $h_{ij}(g_i,g_j, \bm{z}_{ij}) < 0$ holds even if two ellipsoids are collision-free. 
In other words, a naive choice of $\bm{z}_{ij}$ makes $h_{ij}(g_i,g_j, \bm{z}_{ij})$ smaller than the actual distance $w_{ij}^*$, resulting in too conservative evasive motion. 
To alleviate this gap, let us develop the following optimization problem that intends to maximize $h_{ij}(g_i,g_j, \bm{z}_{ij})$ by rotating the supporting hyperplane on the boundary of $\mc E_i$.
\begin{subequations}\label{eq:dual}
\begin{align}
    \max_{\bm{z}_{ij}}~& h_{ij}(g_i,g_j, \bm{z}_{ij})\\
    \mbox{s.t.}~& \| \bm{z}_{ij} \| = 1
\end{align}
\end{subequations}
Notice that the optimization problem \eqref{eq:dual} maximizes $h_{ij}(g_i,g_j, \bm{z}_{ij})$ over the unit vector $\bm{z}_{ij}$ while fixing two ellipsoids. 
In the rest of this subsection, let us denote $h_{ij}(g_i,g_j, \bm{z}_{ij})$ as $h_{ij}(\bm{z}_{ij})$ for notational simplicity since we discuss the property of the optimization problem \eqref{eq:dual} that presumes fixed ellipsoids.
Then, in the next subsection, we discuss how to update $\bm{z}_{ij}$ with moving ellipsoidal rigid bodies to establish a real-time collision avoidance methodology.

To elucidate the meaning of the optimization problem \eqref{eq:dual}, let us introduce the optimization problem
\begin{subequations}\label{eq:primal}
\begin{align}
    \min_{\bm{x},\bm{y},\bm{w}}~&\|\bm{w}\|\\
    \mbox{s.t.}~&f_i(\bm{x})\leq 0,~f_j(\bm{y})\leq 0,\\
    &\bm{y}-\bm{x}=\bm{w},
\end{align}
\end{subequations}
where $f_i(\bm{x}):=(\bm{x}-\bm{p}_i)^\top\bar{Q}_i^{-2}(\bm{x}-\bm{p}_i)-1\leq 0$ signifies the condition $\bm{x} \in \mc E_i$.
Therefore, the problem \eqref{eq:primal} returns the shortest distance between two ellipsoids, namely $w_{ij}^*$, as its optimal solution $\|\bm{w}^*\|$.
Then, the following theorem formalizes the relationship between $w_{ij}^*$ and $h_{ij}(\bm{z}_{ij})$.
\begin{theorem} \label{th:Prim_dual}
Suppose that two ellipsoids $\mc E_i$ and $\mc E_j$ have no overlap, namely $\mc E_i \cap \mc E_j = \emptyset$ holds.
Then, the optimization problem \eqref{eq:dual} is the dual of the problem \eqref{eq:primal}.
Furthermore, the strong duality holds between the optimization problems \eqref{eq:primal} and \eqref{eq:dual}, namely the following condition holds.
\begin{align} \label{eq:no_gap}
    w_{ij}^* = h_{ij}(\bm{z}_{ij}^*) \geq h_{ij}(\bm{z}_{ij})
\end{align}
\end{theorem}
\begin{proof}
See Appendix~\ref{Ap:Prim_dual}. 
Note that our preliminary work \cite{Nishimoto22} proves the case of $d=2$. The proof in Appendix~\ref{Ap:Prim_dual} modifies the one in \cite{Nishimoto22} to suit the formulation of this paper.
\end{proof}

Theorem~\ref{th:Prim_dual} signifies that the proposed update law, maximizing $h_{ij}(\bm{z}_{ij})$ by rotating the supporting hyperplane $l_{ij}$ as in Fig.~\ref{fig:ellipse_rot}, renders $h_{ij}(\bm{z}_{ij})$ the actual distance $w_{ij}^*$ between two ellipsoids.
In addition, the equation \eqref{eq:no_gap} implies $h_{ij}(\bm{z}_{ij})>0$ serves as a sufficient condition for avoiding collisions, even if $\bm{z}_{ij}$ does not converge to the maximizer of \eqref{eq:dual}, $\bm{z}_{ij}^*$.

\subsection{CBFs Integrating Rotating Supporting Hyperplanes} \label{ssec:Propose_method}

This subsection first presents the collision avoidance method for two ellipsoidal rigid bodies. 
We then extend the result to allow distributed implementation for $n$ rigid bodies, assuming each rigid body can communicate its state. 
Note that we again regard $h_{ij}$ as a function of $g_i$, $g_j$, and $\bm{z}_{ij}$, although we have omitted the dependency of $h_{ij}$ for notational simplicity.


As presented in the optimization problem \eqref{eq:dual} and Fig.~\ref{fig:ellipse_rot}, the evaluation of $h_{ij}$ requires a supporting hyperplane $l_{ij}$ between rigid bodies~$i$ and $j$. 
Without loss of generality, let us suppose that a rigid body with a lower ID equips a supporting hyperplane. 
Since the supporting hyperplane should be rotated to minimize the gap between $h_{ij}$ and $w_{ij}^*$, we regard $\bm{z}_{ij}$ as an additional state variable of rigid body $i~(i < j)$. 
In other words, to avoid collisions between rigid bodies $i$ and $j$, we need to control $(g_i, g_j, \bm{z}_{ij})$.

In our proposed method, $\bm{z}_{ij}$ is updated based on the following dynamics
\begin{align} \label{eq:dyn_v} %
    \dot{\bm{z}}_{ij} = \left(I_d -\bm{z}_{ij} \bm{z}_{ij}^\top\right) \bm{u}_{\bm{z}_{ij}} 
\end{align}
with the input $\bm{u}_{\bm{z}_{ij}}\in \R^d$ rendered by a QP for collision avoidance, as will be presented in \eqref{eq:QP_two}.
Note that the update rule \eqref{eq:dyn_v} restricts $\bm{z}_{ij}$ on the unit ball.
Since the pose of rigid bodies $i$ and $j$ follow the dynamics shown in \eqref{eq:dyn_RBM}, the combined control input to be considered for two rigid bodies $i$ and $j$ can be denoted as $\bm{u}_{ij} = [ \bm{v}_i^\top~\bm{\omega}_i^\top~\bm{u}_{\bm{z}_{ij}}^\top~\bm{v}_j^\top~\bm{\omega}_j^\top]^\top$.

Let us denote the nominal input of $\bm{z}_{ij}$ as $\bm{u}_{\mathrm{nom},\bm{z}_{ij}}$ that will be integrated into the CBF-based framework, as discussed in Section~\ref{sec:Preliminary}. 
Then, $\bm{u}_{\mathrm{nom},\bm{z}_{ij}}$ that intends to maximize $h_{ij}$ can be derived from the gradient ascent law as
\begin{align} \label{eq:grad_ascent}
    \bm{u}_{\mathrm{nom},\bm{z}_{ij}}=k_{\bm z} \frac{\partial h_{ij}}{\partial \bm{z}_{ij}},~~~k_{\bm z} >0,
\end{align}
which drives $h_{ij}$ to the local maximum. 
Note that the maximum point of $h_{ij}$ changes as two rigid bodies move around during the operation. 
Hence, the update rule should make $h_{ij}$ converge fast enough to keep up with their motion. 
This requirement is easily satisfied by setting a large enough $k_{\bm z}$ since $\bm{z}_{ij}$ is a virtual variable that does not depend on the physical dynamics of the rigid bodies.
The calculated result of $\partial h_{ij} / \partial \bm{z}_{ij}$ will be shown in \eqref{eq:diff_results_z}, where we derive $\dot h_{ij}$.


Having defined the state to be considered, we are now ready to present a novel collision avoidance strategy for ellipsoidal rigid bodies. 
Let us introduce a new safe set $\hat{\mc S}_{ij}$ that integrates an angle of the supporting hyperplane, namely $\bm{z}_{ij}$, as follows.
\begin{align}
    \hat{\mc S}_{ij} \!=\! \left\{(g_i, g_j, \bm{z}_{ij}) \!\in\! SE(d) \!\times\! SE(d) \!\times\! \mathbb{S}^{(d-1)}  \,\middle|\, h_{ij}\!\geq\! 0\right\}\!
\end{align}
Since $w_{ij}^*>h_{ij}$ holds from Theorem~\ref{th:Prim_dual}, the original safe set $\mc S_{ij}$ can be rendered forward invariant if the proposed CBF guarantees the forward invariance of $\hat{\mc S}_{ij}$. 
To achieve this goal, we propose the following QP, which synthesizes the nominal input and the proposed CBF $h_{ij}$ as 
\begin{subequations} \label{eq:QP_two}
\begin{align}
    &\bm{u}_{ij}^* = \argmin_{\bm{u}_{ij}}~\left\|\bm{u}_{ij}-\bm{u}_{\mathrm{nom},ij}\right\|^2\\
    &\mbox{s.t.}~\dot{h}_{ij} \geq -\alpha(h_{ij}). \label{eq:trueCBFcond}
\end{align}
\end{subequations}
The following lemma presents $\dot{h}_{ij}$ in both ${d=3}$ and $d=2$.
\begin{lemma} \label{lem:dot_h}
The time-derivative of the CBF~\eqref{eq:CBF_cand} along with the dynamics represented as RBM~\eqref{eq:dyn_RBM} in $d=3$ is derived as
\begin{align} \label{eq:dot_h}
\begin{split}
    \dot{h}_{ij} &= \bm{\zeta}_{ij} R_i \bm{\omega}_i + \bm{\eta}_{ij} R_i \bm{v}_i 
    \!+\! \bm{\mu}_{ij} \left(I_d-\bm{z}_{ij} \bm{z}_{ij}^\top \right) \bm{u}_{z_{ij}} \\
    &+ \bm{\nu}_{ij} R_j \bm{\omega}_j + \bm{\xi}_{ij} R_j \bm{v}_j,
\end{split}
\end{align}
where each coefficient can be derived as
\begin{subequations} \label{eq:diff_results}
\begin{align}
\begin{split}
&\bm{\zeta}_{ij} = \frac{\rho}{\left\|\bar{Q}_i^{-1}\bm{z}_{ij}\right\|^3} \bm{z}_{ij}^\top \bar{Q}_i^{-2} \hat{\bm{z}}_{ij} \\
&+\! \frac{1}{\sigma}\left( \bm{z}_{ij}^\top \bar{Q}_i^{-1} \bar{Q}_j^{2} \right) \left( \left(\bar{Q}_i^{-1} \bm{z}_{ij} \right)^\wedge - \bar{Q}_i^{-1} \hat{\bm{z}}_{ij}  \right) \\
&+\! \frac{1}{\left\|\bar{Q}_i^{-1} \! \bm{z}_{ij}\right\|} \! \left(\! \left( \bm{p}_j \!-\! \bm{p}_i \right)^\top \! \bar{Q}_i^{-1} \! \hat{\bm{z}}_{ij} \!-\! \bm{z}_{ij}^\top \bar{Q}_i^{-1} \! \left( \bm{p}_j \!-\! \bm{p}_i \right)^\wedge \! \right)\!,
\end{split} \label{eq:diff_results_omegai} \\
&\bm{\eta}_{ij} = -\frac{1}{\left\|\bar{Q}_i^{-1}\bm{z}_{ij}\right\|} \bm{z}_{ij}^\top \bar{Q}_i^{-1}, \label{eq:diff_results_vi} \\
\begin{split} \label{eq:diff_results_z}
&\bm{\mu}_{ij} \!=\! \frac{\rho}{\left\|\bar{Q}_i^{-1}\bm{z}_{ij}\right\|^3} \bm{z}_{ij}^\top \bar{Q}_i^{-2} 
\!+\! \frac{1}{\left\|\bar{Q}_i^{-1}\bm{z}_{ij}\right\|} \left( \bm{p}_j \!-\! \bm{p}_i \right)^\top \!\! \bar{Q}_i^{-1} \\
&- \frac{1}{\sigma} \bm{z}_{ij}^\top \bar{Q}_i^{-1} \bar{Q}_j^{2}  \bar{Q}_i^{-1}, 
\end{split} \\
&\bm{\nu}_{ij}  = -\frac{1}{\sigma} \bm{z}_{ij}^\top \bar{Q}_i^{-1} \bar{Q}_j^{2} \left( \bar{Q}_i^{-1} \bm{z}_{ij} \right)^\wedge, \label{eq:diff_results_omegaj} \\
&\bm{\xi}_{ij} = \frac{1}{\left\|\bar{Q}_i^{-1}\bm{z}_{ij}\right\|} \bm{z}_{ij}^\top \bar{Q}_i^{-1}, \label{eq:diff_results_vj}
\end{align}
\end{subequations}
with
\begin{subequations} \label{eq:diff_results_const}
\begin{align} 
    \rho &= \left( 1- \left( \bm{p}_j-\bm{p}_i \right)^\top \bar{Q}_{i}^{-1}\bm{z}_{ij} + \left\|\bar{Q}_j \bar{Q}_i^{-1} \bm{z}_{ij}\right\| \right), \\
    \sigma &= \left\|\bar{Q}_j \bar{Q}_i^{-1} \bm{z}_{ij}\right\| \left\|\bar{Q}_i^{-1}\bm{z}_{ij}\right\|.
\end{align}
\end{subequations}
Furthermore, $\dot h_{ij}$ in $d=2$ is calculated as 
\begin{align} \label{eq:dot_h_2D}
\begin{split}
    \dot{h}_{ij} &= \tilde{\bm{\zeta}}_{ij} \bm{\omega}_i + \bm{\eta}_{ij} R_i \bm{v}_i 
    \!+\! \bm{\mu}_{ij} \left(I_d-\bm{z}_{ij} \bm{z}_{ij}^\top\right) \bm{u}_{z_{ij}} \\
    &+ \tilde{\bm{\nu}}_{ij} \bm{\omega}_j + \bm{\xi}_{ij} R_j \bm{v}_j,
\end{split}
\end{align}
with
\begin{subequations} \label{eq:diff_results_2D}
\begin{align}
    \begin{split} 
    &\tilde{\bm{\zeta}}_{ij} = -\frac{\rho}{\left\|\bar{Q}_i^{-1}\bm{z}_{ij}\right\|^3} \bm{z}_{ij}^\top \bar{Q}_i^{-2} \hat{1}\bm{z}_{ij} \\
    &-\! \frac{1}{\sigma}\left( \bm{z}_{ij}^\top \bar{Q}_i^{-1} \bar{Q}_j^{2} \right) \left( \hat{1}\bar{Q}_i^{-1} \bm{z}_{ij}  - \bar{Q}_i^{-1} \hat{1}\bm{z}_{ij}  \right) \\
    &-\! \frac{1}{\left\|\bar{Q}_i^{-1}  \bm{z}_{ij}\right\|} \! \left(\! \left( \bm{p}_j \!-\! \bm{p}_i \right)^\top \! \bar{Q}_i^{-1} \hat{1}\bm{z}_{ij} + \bm{z}_{ij}^\top \bar{Q}_i^{-1} \hat{1}\left( \bm{p}_j \!-\! \bm{p}_i \right) \! \right),
    \end{split} \label{eq:diff_results_omegai_2D} \\
    &\tilde{\bm{\nu}}_{ij}  = \frac{1}{\sigma} \bm{z}_{ij}^\top \bar{Q}_i^{-1} \bar{Q}_j^{2} \hat{1} \left( \bar{Q}_i^{-1} \bm{z}_{ij} \right) \label{eq:diff_results_omegaj_2D}.
\end{align}
\end{subequations}
\end{lemma}
\begin{proof}
See Appendix~\ref{Ap:dot_h_calc}.
\end{proof}
Let us emphasize that, as stated in the contributions in Section~\ref{ssec:contri}, this paper newly derives $\dot h_{ij}$ along with the dynamics expressed as RBM~\eqref{eq:dyn_RBM} to develop a collision avoidance controller for both 2D and 3D environments.
Note that the difference between \eqref{eq:dot_h} and \eqref{eq:dot_h_2D} mainly stems from the operator $\wedge$ that returns a slightly different result in $d=3$ and $d=2$, as in \eqref{eq:wedge}.
Appendix~\ref{Ap:dot_h_calc} also provides more detailed explanations of this difference.
The discussion in the rest of this section proceeds with the QP \eqref{eq:QP_two} with $\dot h_{ij}$ in \eqref{eq:dot_h} because the results are consistent for both $d=2$ and $d=3$.


Although the QP \eqref{eq:QP_two} rectifies the nominal input to avoid collisions, \eqref{eq:trueCBFcond} cannot be evaluated in a distributed fashion because $\dot h_{ij}$ in \eqref{eq:dot_h} requires the information of both rigid bodies' control inputs and a supporting hyperplane. Nevertheless, the proposed method can be distributed by assuming the communication between rigid bodies as follows.
\begin{assumption} \label{ass:commu}
    Rigid body~$i\in \mc N$ can acquire the pose and the shape information of other rigid bodies, namely $(\bm{p}_j, R_j)$ and $Q_j,~\forall j\in \mc N\backslash \{i\}$. In addition, rigid body~$j$ can receive $\bm{z}_{ij}$ from rigid bodies $i,~\forall i\in \{1 \cdots j-1\}$. 
\end{assumption}
Assumption~\ref{ass:commu} implies each rigid body can acquire other rigid bodies' poses and shapes. 
As presented in \cite{Nicholson2019,Ok2019,Srinivasan20}, it is realizable by equipping some sensors on rigid bodies, such as a visual sensor or LiDAR. 
In contrast, $\bm{z}_{ij}$ with $i<j$ needs to be communicated because $\bm{z}_{ij}$ is a virtual variable of rigid body~$i$ and cannot be estimated by rigid body~$j$ from its sensor data. 
Nevertheless, the dimension of $\bm{z}_{ij}$ is only $d$ and easy to transfer.
Note that although we assume a complete graph in Assumption~\ref{ass:commu} to make a discussion simpler, the proposed method can be extendable to a distance-based graph, such as the $\Delta$-disk proximity graph \cite{mesbahi2010graph}, as conducted in \cite{Ibuki22,Ibuki20}.



With Assumption~\ref{ass:commu}, let us consider evaluating the QP \eqref{eq:QP_two} in a distributed manner, where rigid bodies $i$ and $j$ with $i<j$ are responsible for determining $\bm{u}_i = [\bm{v}_{i}^\top~\bm{\omega}_i^\top~\bm{u}_{\bm{z}_{ij}}^\top ]^\top \in \R^{2d+\frac{d(d-1)}{2}}$ and $\bm{u}_j = [\bm{v}_{j}^\top~\bm{\omega}_j^\top ]^\top \in \R^{d+\frac{d(d-1)}{2}}$, respectively. 
Then, the CBF condition \eqref{eq:trueCBFcond} and \eqref{eq:dot_h} can be decomposed into the following two conditions.
\begin{subequations}
\begin{align}
    &\bm{\zeta}_{ij} R_i \bm{\omega}_i + \bm{\eta}_{ij} R_i \bm{v}_i + \bm{\mu}_{ij} \left(I_d-\bm{z}_{ij} \bm{z}_{ij}^\top \right) \bm{u}_{z_{ij}}
    \geq -\frac{1}{2}\alpha(h_{ij}) \label{eq:CBF_cond_i} \\
    &\bm{\nu}_{ij} R_j \bm{\omega}_j + \bm{\xi}_{ij} R_j \bm{v}_j \geq -\frac{1}{2}\alpha(h_{ij}) \label{eq:CBF_cond_j}
\end{align}
\end{subequations}
Note that both rigid bodies $i$ and $j$ can calculate $h_{ij}$ and coefficients as they have all the information needed to evaluate \eqref{eq:CBF_cand} and \eqref{eq:diff_results} under Assumption~\ref{ass:commu}.
Therefore, the conditions \eqref{eq:CBF_cond_i} and \eqref{eq:CBF_cond_j} can be evaluated by rigid bodies~$i$ and $j$, respectively.
In summary, the QP \eqref{eq:QP_two} is distributed with the following two QPs. 
\begin{subequations} \label{eq:QP_dist_i}
\begin{align}
    &\bm{u}_{i}^* = \argmin_{\bm{u}_{i}}~\left\|\bm{u}_{i}-\bm{u}_{\mathrm{nom},i}\right\|^2\\
    &\mbox{s.t.}~\bm{\zeta}_{ij} R_i \bm{\omega}_i \!+\! \bm{\eta}_{ij} R_i \bm{v}_i \!+\! \bm{\mu}_{ij} \left(I_d-\bm{z}_{ij} \bm{z}_{ij}^\top \right) \bm{u}_{z_{ij}}
    \!\geq\! -\frac{1}{2}\alpha(h_{ij}) \label{eq:QP_cond_dist_i}
\end{align}
\end{subequations}
for rigid body $i$ and 
\begin{subequations} \label{eq:QP_dist_j}
\begin{align}
    &\bm{u}_{j}^* = \argmin_{\bm{u}_{j}}~\left\|\bm{u}_{j}-\bm{u}_{\mathrm{nom},j}\right\|^2\\
    &\mbox{s.t.}~\bm{\nu}_{ij} R_j \bm{\omega}_j + \bm{\xi}_{ij} R_j \bm{v}_j \geq -\frac{1}{2}\alpha(h_{ij}) \label{eq:QP_cond_dist_j}
\end{align}
\end{subequations}
for rigid body $j$.
Note that each QP is only responsible for calculating the control input of one rigid body. This is in contrast to the QP \eqref{eq:QP_two}, which calculates both $\bm{u}_{i}$ and $\bm{u}_{j}$ in a centralized manner.

The existence of the control input satisfying the constraint in \eqref{eq:QP_cond_dist_i} and \eqref{eq:QP_cond_dist_j}, also called as the validity of the CBF \cite{Ames2019_CBF_thapp,Notomista21}, is guaranteed by the following theorem.
\begin{theorem} \label{thm:CBF_valid}
The function $h_{ij}$ in \eqref{eq:CBF_cand} is a valid CBF with the distributed QPs in \eqref{eq:QP_dist_i} and \eqref{eq:QP_dist_j} for both $d=2$ and $d=3$ when $\bm{u}_i \in \R^{2d+\frac{d(d-1)}{2}},~\bm{u}_j\in \R^{d+\frac{d(d-1)}{2}}$. 
\end{theorem}
\begin{proof}
Let us first discuss the case of $d=3$. 
To guarantee both QPs \eqref{eq:QP_dist_i} and \eqref{eq:QP_dist_j} have a solution for any $(\bm{p}_i, R_i), (\bm{p}_j, R_j) \in SE(3)$ and $\bm{z}_{ij} \in {\mathbb S}^{2}$, we need to prove the coefficients of the constraints in \eqref{eq:QP_cond_dist_i} and \eqref{eq:QP_cond_dist_j} never become zero vecotors, namely $[\bm{\zeta}_{ij}^\top~\bm{\eta}_{ij}^\top~\bm{\mu}_{ij}^\top]^\top \neq \bm{0}$ and $[\bm{\nu}_{ij}^\top~\bm{\xi}_{ij}^\top]^\top \neq \bm{0}$. 
These two conditions are always satisfied because $\| \bm{\eta}_{ij} \| = \| \bm{\xi}_{ij} \| =1$ holds from \eqref{eq:diff_results_vi} and \eqref{eq:diff_results_vj}.
Therefore, each QP can find a solution satisfying the constraints by appropriately selected $\bm{v}_i$ and $\bm{v}_j$. 
The above discussion also applies with $d=2$ because $\dot{h}_{ij}$ in $d=2$ also has $\bm{\eta}_{ij}$ and $\bm{\xi}_{ij}$ as the coefficients of $\bm{v}_i$ and $\bm{v}_j$, as shown in \eqref{eq:dot_h_2D}. This completes the proof.
\end{proof}
Theorem~\ref{thm:CBF_valid} means that, for any $(\bm{p}_i, R_i), (\bm{p}_j, R_j) \in SE(d)$ and $\bm{z}_{ij} \in {\mathbb S}^{(d-1)}$, both QPs \eqref{eq:QP_dist_i} and \eqref{eq:QP_dist_j} can find the control input that renders the set $\hat{\mc S}_{ij}$ forward invariant.

The above discussion can be easily extended to the case of $n$ rigid bodies. Because the rigid body with a smaller ID among any two rigid bodies has a supporting hyperplane, rigid body $i$ owns $(n-i)$ supporting hyperplanes. Let us introduce the vector combining states of all supporting hyperplanes rigid body~$i$ has as $\bm{z}_{i}=[\bm{z}_{i\hspace{0.4mm}i+1}^\top~\bm{z}_{i\hspace{0.4mm}i+2}^\top \cdots \bm{z}_{i\hspace{0.4mm}n}^\top]^\top$ with $\bm{z}_n=\emptyset$, which has the following dynamics as \eqref{eq:dyn_v} 
\begin{align} \label{eq:dyn_zi}
    \dot{\bm{z}}_i = \left( I_{(n-i)} \otimes (I_d -\bm{z}_{ij} \bm{z}_{ij}^\top) \right) \bm{u}_{\bm{z}_i},
\end{align}
where $\otimes$ represents the Kronecker product and $\bm{u}_{\bm{z}_i} = [\bm{u}_{\bm{z}_{i\hspace{0.4mm}i+1}}^\top \cdots \bm{u}_{\bm{z}_{i\hspace{0.4mm}n}}^\top]^\top$.
Then, the ensemble state of rigid body~$i$ becomes $X_i = (\bm{p}_i, R_i, \bm{z}_{i})$.
Rigid body~$i$ calculates the control input $\bm{u}_i = [\bm{v}_{i}^\top~\bm{\omega}_i^\top~\bm{u}_{\bm{z}_{i}}^\top ]^\top$ for this ensemble state $X_i$ by solving the QP
\begin{subequations} \label{eq:QP_N}
\begin{align}
    &\bm{u}_{i}^* = \argmin_{\bm{u}_{i}}~\left\|\bm{u}_{i}-\bm{u}_{\mathrm{nom},i}\right\|_W^2\\
\begin{split}
    &\mbox{s.t.}~\bm{\zeta}_{ij} R_i \bm{\omega}_i + \bm{\eta}_{ij} R_i \bm{v}_i + \bm{\mu}_{ij} (I_d-\bm{z}_{ij} \bm{z}_{ij}^\top) \bm{u}_{\bm{z}_{ij}} \\
    &\hspace{3.4cm}\geq -\frac{1}{2}\alpha(h_{ij}),~ \forall j > i,
\end{split} \\
    &\hspace{0.4cm}~\bm{\nu}_{ij} R_j \bm{\omega}_j + \bm{\xi}_{ij} R_j \bm{v}_j \geq -\frac{1}{2}\alpha(h_{ij}),~ \forall j < i,
\end{align}
\end{subequations}
where $\|\bm{a}\|_W = \bm{a}^\top W \bm{a}$ with a diagonal matrix $W = \diag(\beta_{\bm{v}}I_d, \beta_{\bm{\omega}}I_{d(d-1)/2}, I_{d(n-i)})$ that adjusts the unit difference and can also reflect a user preference on which translational and angular velocity is modified more by $\beta_{\bm{v}}, \beta_{\bm{\omega}}>0$.

%% file: text/5-Simulation.tex
\section{Simulation Results with Rigid Body Motion} \label{sec:examples}

This section presents simulation studies, where their movies, including the one in Section~\ref{sec:vehicle}, can be found in the supplementary material.
The proposed method only requires solving \eqref{eq:QP_N} to generate the collision-free input for each rigid body. The computational time of \eqref{eq:QP_N} is short, within 1-2\,ms in most cases for all simulation studies, where we utilize \textit{quadprog} in MATLAB and a PC equipped with Intel i7-10700 CPU.


\subsection{Two-Dimensional Case}

This subsection presents simulation results with two elliptical rigid bodies in a 2D environment, which demonstrates the proposed CBF and the input for the supporting hyperplane successfully prevent collisions between two elliptical rigid bodies. 
The length of the major and the minor axis of two rigid bodies are randomly selected from the range from $1.0$\,m to $2.0$\,m and from $0.4$\,m to $0.8$\,m, respectively.
The initial positions of the two rigid bodies are $\bm{p}_1=[-3~-3]^\top$ and $\bm{p}_2=[2~0]^\top$, as shown in Fig.~\ref{fig:sim_2Dsimp_snap}\subref{subfig:sim_2Dsimp_00s}. 
Figure~\ref{fig:sim_2Dsimp_snap}\subref{subfig:sim_2Dsimp_00s} also illustrates the supporting hyperplane incorporated in the proposed CBF as a black line together with the shortest line connecting the supporting hyperplane and the other rigid body, namely its length corresponding with $h_{12}$. 
We employ $\alpha(h) = 10h$ as an extended class $\mc K$ function.
We set $\beta_{\bm v} = 1$ and $\beta_{\bm{\omega}} = 0.5$.
The nominal input to $\bm{z}_{12}$ is set to $\bm{u}_{{\rm nom}, \bm{z}_{12}} = 20 (\partial h_{12} / \partial \bm{z}_{12})$.
Note that the initial value of $\bm{z}_{12}$ for the supporting hyperplane is set randomly within a range that separates two rigid bodies.

The snapshots of the simulation are shown in Fig.~\ref{fig:sim_2Dsimp_snap}, where a supporting hyperplane rotates on rigid body~1, adapting to the motion of the rigid bodies. 
As a result, the length of the green line, $h_{12}$, takes almost the same value as the actual distance between rigid bodies, $w_{12}^*$. Figure~\ref{fig:sim_2D_simp_CBF} depicts the evolution of $h_{12}$, shown as a blue line, together with the actual distance $w_{12}^*$ calculated from the optimization problem \eqref{eq:primal}, as a red dashed line. 
Note that $w_{12}^*$ is calculated just for comparison, and the proposed method does not necessitate solving \eqref{eq:primal}.
By comparing $h_{12}$ and $w_{12}^*$ in Fig.~\ref{fig:sim_2D_simp_CBF}, $h_{12}$ follows $w_{12}^*$ fast enough to eliminate the conservativeness, namely the error between $h_{12}$ and $w_{12}^*$. In addition, because $h_{12}$ is always larger than $0$, collision avoidance is successfully achieved.

\begin{figure}[t!]
    \vspace{-0.2cm}
    \centering
    \subfloat[$t=0$\,s]
    {\makebox[0.48\hsize][c]{\includegraphics[trim = 0.3cm 0.6cm 0.3cm 1.5cm, clip=true, width=0.48\linewidth]{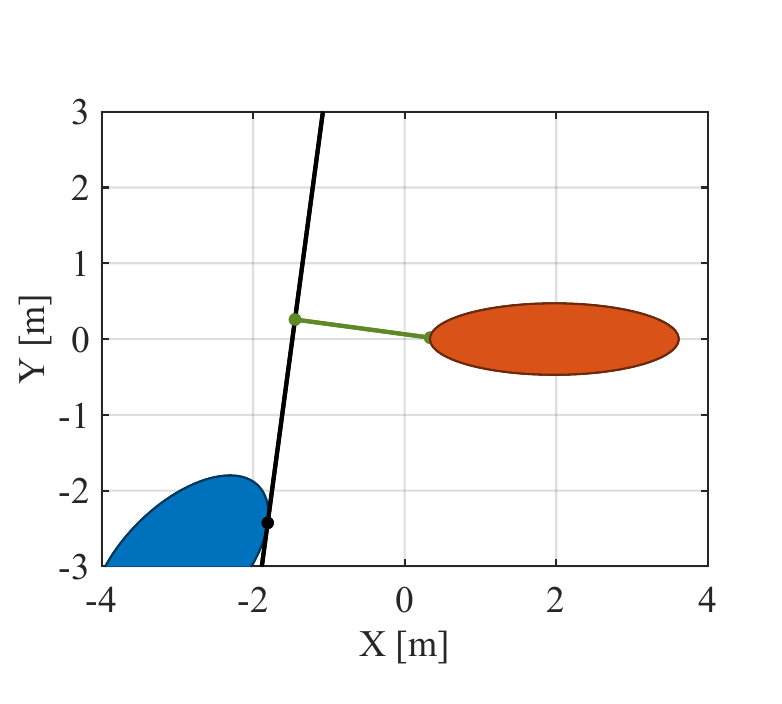}
    \label{subfig:sim_2Dsimp_00s}}} 
    \subfloat[$t=0.8$\,s]
    {\makebox[0.48\hsize][c]{\includegraphics[trim = 0.3cm 0.6cm 0.3cm 1.5cm, clip=true, width=0.48\linewidth]{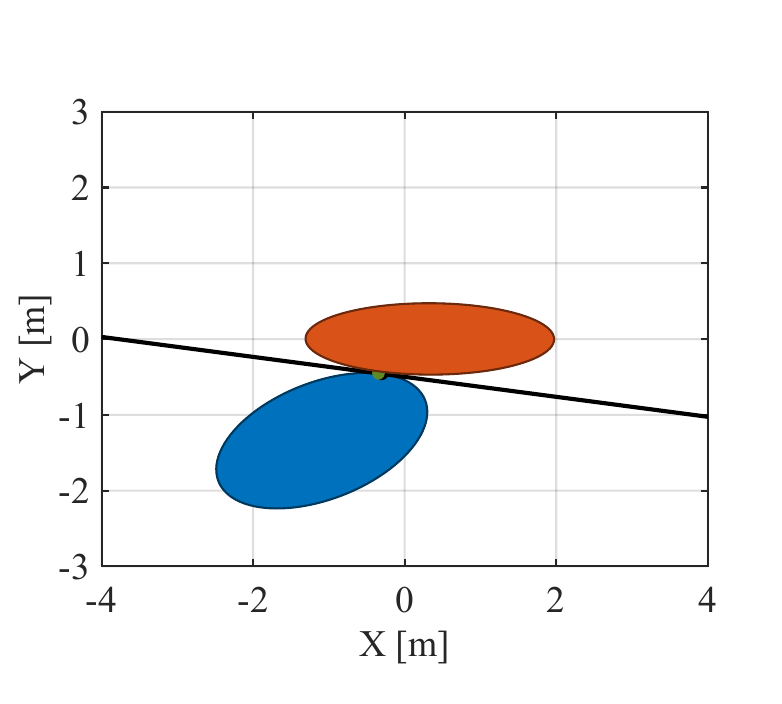}
    \label{subfig:sim_2Dsimp_08s}}}\\
    \subfloat[$t=1.8$\,s]
    {\makebox[0.48\hsize][c]{\includegraphics[trim = 0.3cm 0.6cm 0.3cm 1.5cm, clip=true, width=0.48\linewidth]{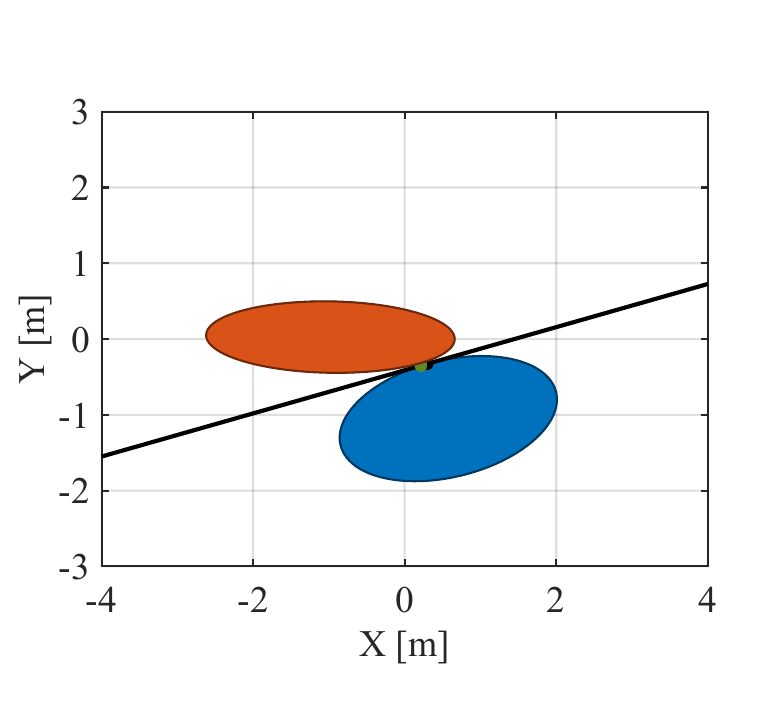}
    \label{subfig:sim_2Dsimp_18s}}} 
    \subfloat[$t=3.5$\,s]
    {\makebox[0.48\hsize][c]{\includegraphics[trim = 0.3cm 0.6cm 0.3cm 1.5cm, clip=true, width=0.48\linewidth]{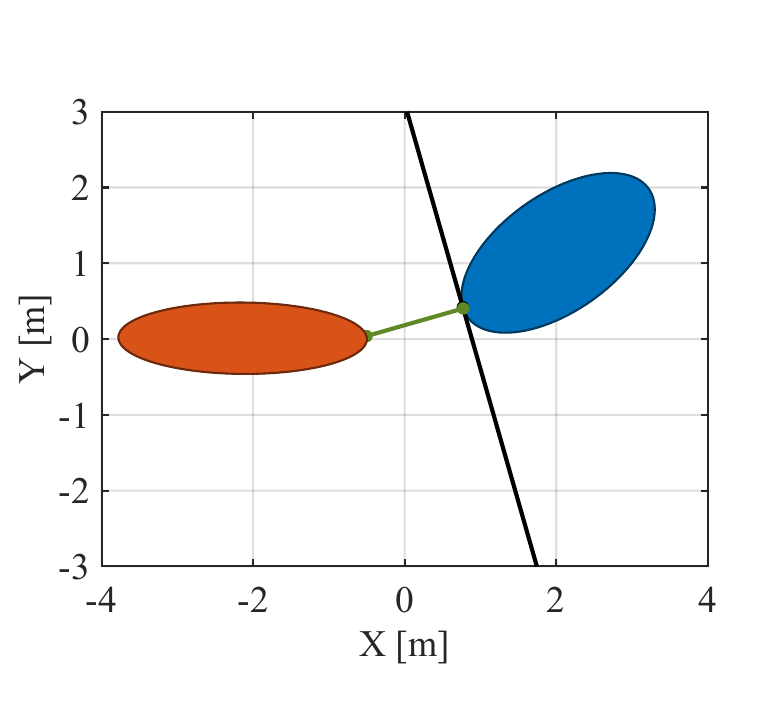}
    \label{subfig:sim_2D_35s}}} 
    \caption{Snapshots of the simulation with the proposed collision avoidance method, where rigid body~1 depicted in blue has a supporting hyperplane. The supporting hyperplane is rotated to maximize the signed distance $h_{12}$, which corresponds with the length of the green line.}
    \label{fig:sim_2Dsimp_snap}
\end{figure}

\begin{figure}
    \centering
    \includegraphics[width=50mm]{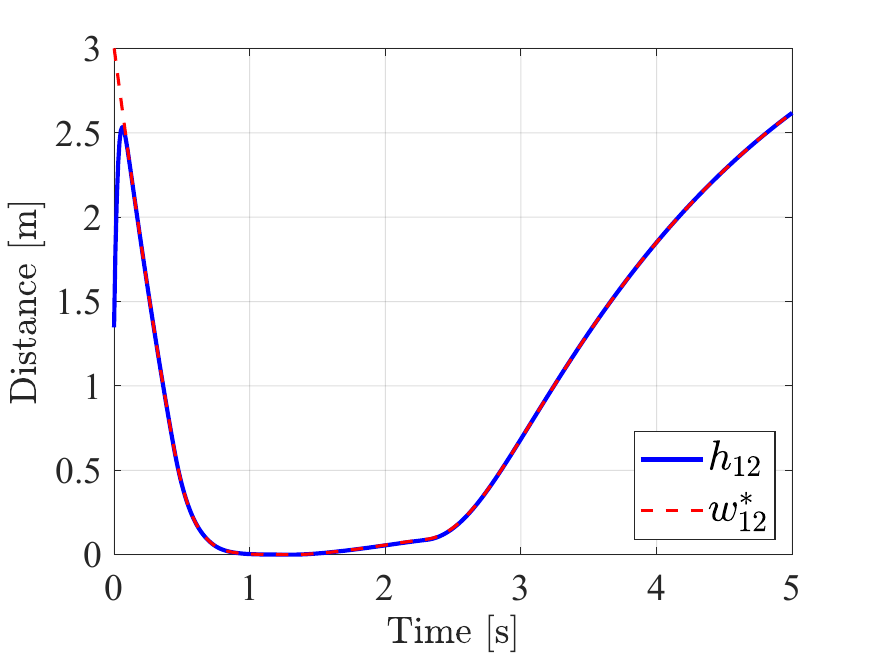}
    \caption{Evolution of the CBF $h_{12}$ and the actual minimum distance between two elliptical rigid bodies, $w_{12}^*$. The CBF $h_{12}$ always takes a smaller value than that of $w_{12}^*$ while keeping a close value with it.}
    \label{fig:sim_2D_simp_CBF}
\end{figure}

\subsection{Comparative Simulation in Two-Dimension}
This subsection provides the comparative study between our proposed CBF and the one employing the discriminant of a cubic characteristic equation \cite{Ibuki22}, which showcases the proposed method exhibits more favorable avoidance trajectories.
To make it easier to understand the difference, this subsection considers a scenario, where an elliptical robot avoids a fixed elliptical obstacle.
If we index the elliptical robot and the obstacle as rigid bodies $1$ and $2$, the QP to be solved by rigid body~1, namely, the robot, is written as 
\begin{subequations} \label{eq:QP_two_obstacle}
\begin{align}
    &\bm{u}_{1}^* = \argmin_{\bm{u}_{1}}~\left\|\bm{u}_{1}-\bm{u}_{\mathrm{nom},1}\right\|_W^2\\
    &\mbox{s.t.}~\dot{h}_{12} \geq -\alpha(h_{12}), 
\end{align}
\end{subequations}
with
\begin{align} \label{eq:dot_h_obstacle}
\begin{split}
    \dot{h}_{12} &= \tilde{\bm{\zeta}}_{12} \bm{\omega}_1 + \bm{\eta}_{12} R_1 \bm{v}_1 
    \!+\! \bm{\mu}_{12} \left(I_d-\bm{z}_{12} \bm{z}_{12}^\top\right) \bm{u}_{\bm{z}_{12}}. \\
\end{split}
\end{align}
Note that since there are no control inputs for the obstacle, the terms having $\bm{\omega}_2$ and $\bm{v}_2$ in \eqref{eq:dot_h_2D} are removed in \eqref{eq:dot_h_obstacle}.

Let us introduce the CBF presented in \cite{Ibuki22}, with which we will compare the performance of our proposed CBF.
The CBF in \cite{Ibuki22} is based on the following cubic characteristic equation 
\begin{align}
J(x) &= \det\left( x T_i(g_i) -T_j(g_j) \right), \label{eq:ellip_chara}\\
T_i(g_i) &= \left(g_i^{-1}\right)^\top 
\begin{bmatrix}
Q_i^{-2} & 0 \\
0 & -1
\end{bmatrix}
g_i^{-1},
\end{align}
where $g_i$ is defined in \eqref{eq:gi_homoge}.
Let us denote the discriminant of \eqref{eq:ellip_chara} as $D_{ij}(x)$.
The discriminant $D_{ij}(x)$ has two properties: (i) if ellipses $i$ and $j$ do not overlap, $D_{ij}(x)>0$ holds and (ii) if ellipses $i$ and $j$ contact with each other without overlapping, $D_{ij}(x)=0$ holds.
The work \cite{Ibuki22} presents a CBF
\begin{align}
    h_{discr} = D_{ij}(x)
\end{align}
to leverage the above property of $D_{ij}(x)$ for collision avoidance.
For more detail, please see \cite{Choi2005, Ibuki22}.
Note that, differently from the proposed CBF $h_{ij}$, the CBF $h_{discr} = D_{ij}(x)$ cannot be extended to three-dimensional ellipsoidal rigid bodies because it requires solving the higher-dimensional characteristic equations, the solution of which is divided by cases or not obtained in explicit form.

In the following comparative studies, we consider a scenario in which the robot, rigid body~1, moves toward a goal while avoiding an obstacle, rigid body~2. 
The shapes of rigid bodies are modeled as ellipses, which sizes are specified as $Q_1 = \diag(0.5, 0.2)$ and $Q_2=\diag(3, 2)$, respectively. 
The initial and the goal pose of the robot are $(\bm{p}_{\rm ini}, R_{\rm ini}) = ([12~0]^\top, I_2)$ and $(\bm{p}_{\rm goal}, R_{\rm goal}) = ([0~0]^\top, I_2)$. 
The nominal input to $\bm{z}_{12}$ is set to $u_{{\rm nom}, \bm{z}_{12}} = 10 (\partial h_{12} / \partial \bm{z}_{12})$. We set $\beta_{\bm v} = 1$ and $\beta_{\bm{\omega}} = 0.3$ for a weighting matrix $W$.
As an extended class $\mc K$ function, we utilize $\alpha(h) = \gamma h$ for both CBFs, where we analyze how the result changes in each CBF with $\gamma = 1, 10, 100$.


\begin{figure}[t!]
    \centering
    \subfloat[Trajectory]
    {\makebox[0.53\hsize][c]{\includegraphics[trim = 0.3cm 0cm 1cm 0.7cm, clip=true, width=0.53\linewidth]{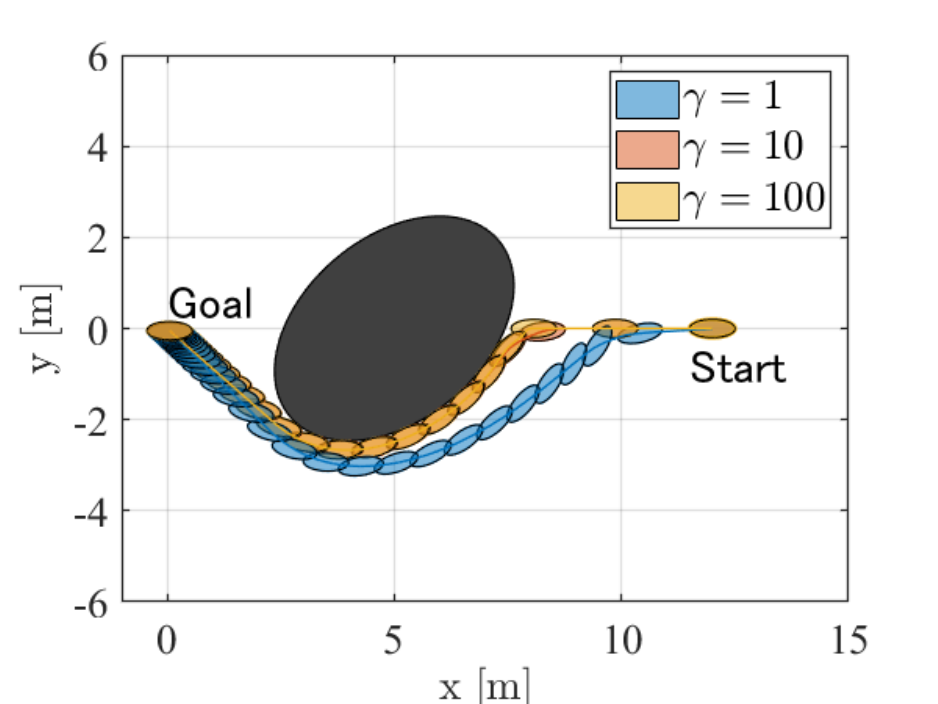}
    \label{subfig:sim_2D_rot_traj}}} 
    \subfloat[Evolution of the proposed CBF]
    {\makebox[0.46\hsize][c]{\includegraphics[trim = 0.2cm 0cm 0.5cm 0.2cm, clip=true, width=0.46\linewidth]{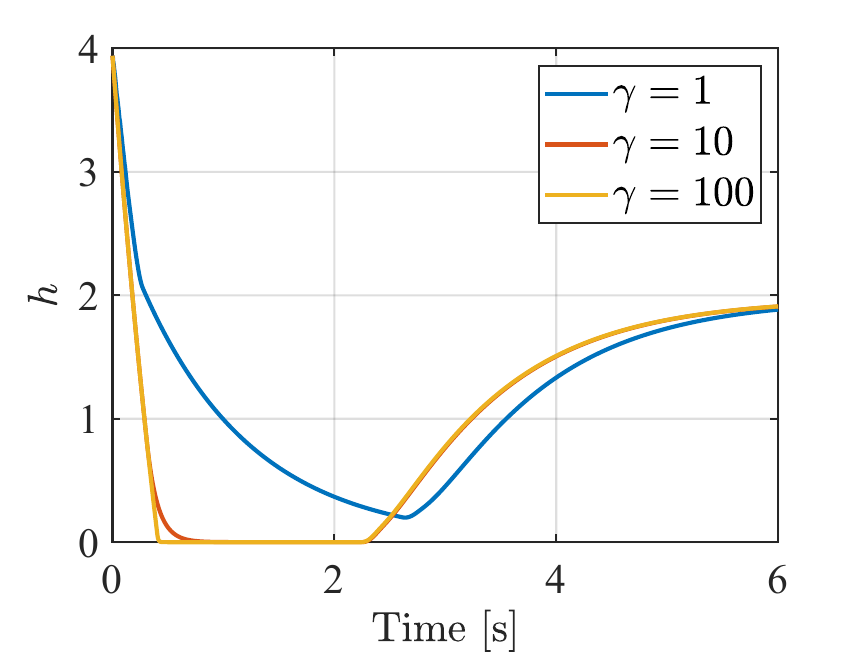}
    \label{subfig:sim_2D_rot_h}}}
    \caption{Simulation results of the proposed CBF with $\gamma =1, 10, 100$. (a): All three cases change their attitudes in the direction smooth avoidance is achieved. (b): The value of the proposed CBF $h$ corresponds with the actual distance between the supporting hyperplane and the obstacle.}
    \label{fig:sim_2D_rot}
\end{figure}

\begin{figure}[t!]
    \centering
    \subfloat[Trajectory]
    {\makebox[0.53\hsize][c]{\includegraphics[trim = 0.3cm 0cm 1cm 0.7cm, clip=true, width=0.53\linewidth]{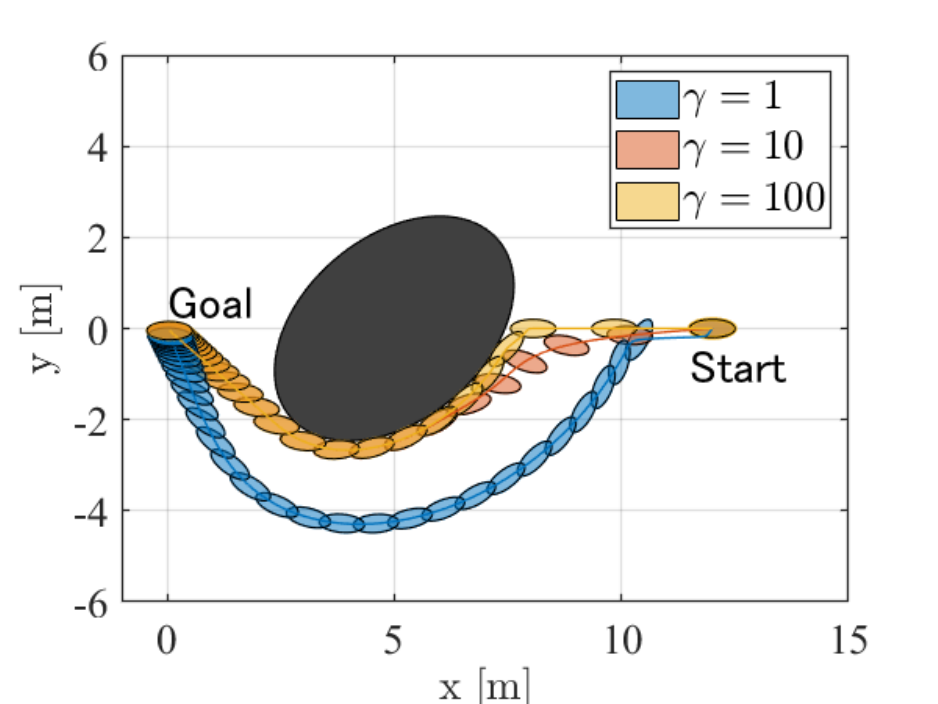}
    \label{subfig:sim_2D_dis_traj}}} 
    \subfloat[Evolution of the CBF in \cite{Ibuki22}]
    {\makebox[0.46\hsize][c]{\includegraphics[trim = 0.2cm 0cm 1cm 0.1cm, clip=true, width=0.46\linewidth]{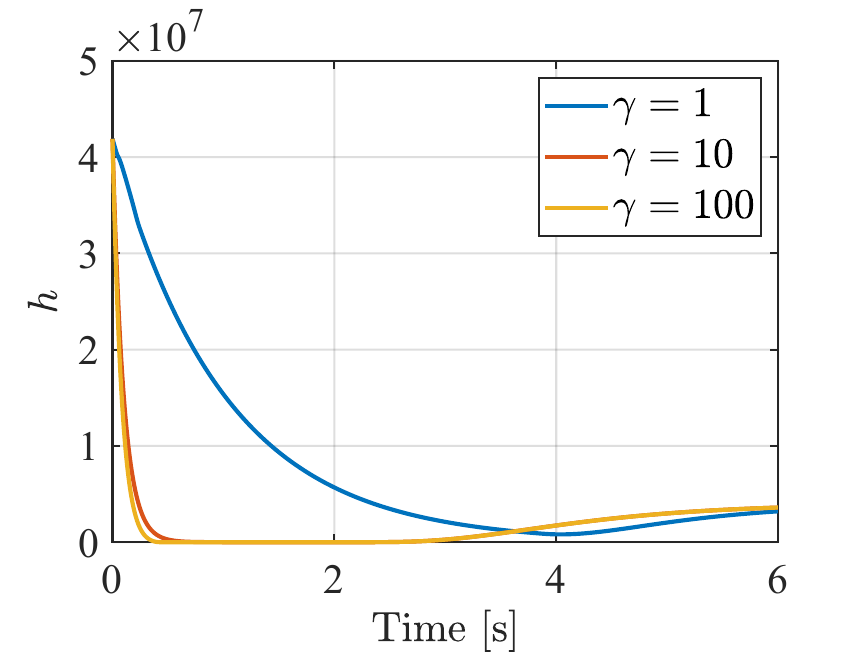}
    \label{subfig:sim_2D_dis_h}}}
    \caption{Simulation results with the CBF proposed in~\cite{Ibuki22}, where the discriminant of the cubic characteristic equation is employed as a CBF. (a): With $\gamma = 10$ shown in red, the robot shows an unreasonable change of its attitude, which does not contribute to securing a margin between the obstacle. The result with $\gamma = 1$ shows a sharper change in its attitude with a larger detour than that of the proposed method in Fig.~\ref{fig:sim_2D_rot}. (b): A CBF in \cite{Ibuki22} does not correspond with the distance, resulting in enormous values.}
    \label{fig:sim_2D_dis}
\end{figure}

Figures~\ref{fig:sim_2D_rot}\subref{subfig:sim_2D_rot_traj} and \ref{fig:sim_2D_dis}\subref{subfig:sim_2D_dis_traj} show the trajectories of the robot with the proposed CBF $h_{12}$ and one with $h_{discr} = D_{12}$ in \cite{Ibuki22}, respectively.
Along with each trajectory, we plot a pose of the robot every 0.2\,s.
In all six trajectories, the robot successfully avoids the collision with the obstacle shown in black. 
Nevertheless, the results of the CBF in \cite{Ibuki22} exhibit unreasonable rotational inputs when the robot approaches the obstacle, as shown in red ellipses in Fig.~\ref{fig:sim_2D_dis}\subref{subfig:sim_2D_dis_traj}. 
More concretely, the robot changes its attitude in the direction not to contribute to collision avoidance. 
This unreasonable change in the attitude appears because the discriminant utilized in $h_{discr}$ does not necessarily have a proportional relationship with the Euclidean distance between ellipses. 
In addition, since the physical interpretation of $h_{discr}$ is difficult to understand, the tuning of parameters is prone to be complicated, as seen in Fig.~\ref{fig:sim_2D_dis}\subref{subfig:sim_2D_dis_traj}, in which all three trajectories vary greatly. 
The evolution of $h_{discr}$, shown in Fig.~\ref{fig:sim_2D_dis}\subref{subfig:sim_2D_dis_h}, also exhibits difficulties in understanding its physical interpretation because its value takes a significantly larger value than that of the distance.
In contrast, the robots with the proposed CBF $h_{12}$ change their attitudes so that the distance from the obstacle increases. 
Figure~\ref{fig:sim_2D_rot}\subref{subfig:sim_2D_rot_h} also verifies that the value of CBF takes almost the same value as the distance.

\subsection{Three-Dimensional Case}

This subsection demonstrates the proposed collision avoidance method \eqref{eq:QP_N} with the rigid body network in a 3D space. 
The rigid body network is composed of sixteen ellipsoidal rigid bodies, eight of which are located at the lower-left corner and the others at the upper-right corner at the initial time, as shown in Fig.~\ref{fig:sim_3D_snap}\subref{subfig:sim_3D_00s}. 
Each group moves toward the other corner of the space while rearranging the positions of the rigid bodies in each group.
The rearranged positions in the goal configuration and the ID of rigid bodies are randomly set.
The attitude of each ellipsoidal rigid body is determined so that its major axis directs to the goal position. 
The length of the major axis of all rigid bodies is set to 1\,m, while the length of the other axes is determined randomly from $0.3$\,m to $0.7$\,m.
With denoting the goal pose for rigid body~$i$ as $(\bm{p}_i^{\rm goal}, R_i^{\rm goal})$, the nominal inputs for the translational and angular body velocities are calculated as $\bm{v}_{{\rm nom}, i} = k_{\bm{v}} R_i^\top ( \bm{p}_i^{\rm goal} - \bm{p}_i )$ and 
${\bm \omega}_{{\rm nom}, i} = \frac{k_{\bm{\omega}}}{2} \left( R_i^\top R_i^{\rm goal} - (R_i^\top R_i^{\rm goal})^\top \right)^\vee$, where $k_{\bm{v}} = 3$, $k_{\bm{\omega}} = 0.5$, and $\vee$ is the inverse operator of $\wedge$ as defined below \eqref{eq:wedge}.
The weighting matrix $W$ is set with $\beta_{\bm{v}}=0.1$ and $\beta_{\bm{\omega}}=1$.
The gain for the nominal input \eqref{eq:grad_ascent} for the supporting hyperplane is set to $k_{\bm z} = 20/(1+h_{ij}^2)$, which is designed to take a smaller value when two ellipsoids are located far away.
The extended class $\mc K$ function is $\alpha(h) = 10h$.

\begin{figure}[t!]
    \vspace{-0.2cm}
    \centering
    \subfloat[$t=0$\,s]
    {\makebox[0.48\hsize][c]{\includegraphics[trim = 0.3cm 0cm 0.5cm 0.7cm, clip=true, width=0.48\linewidth]{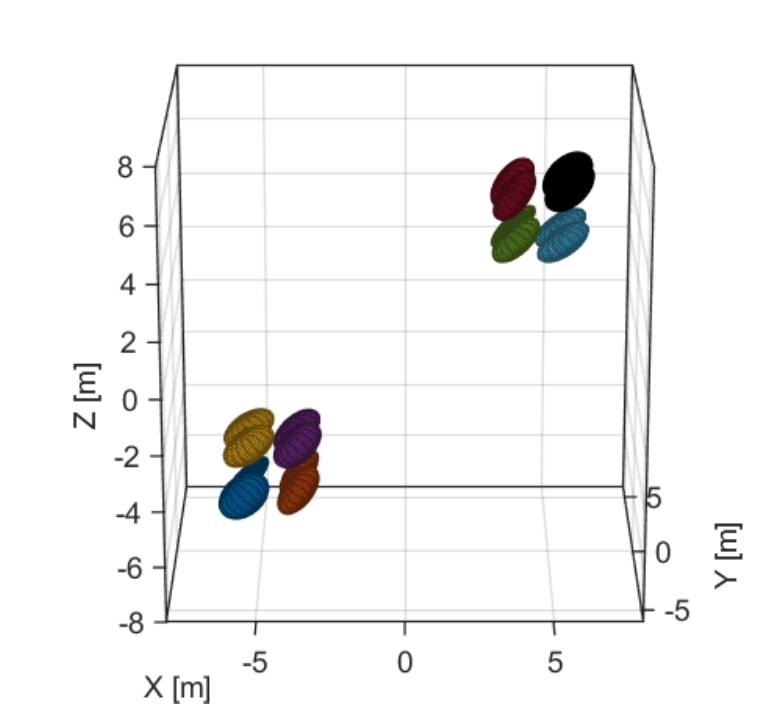}
    \label{subfig:sim_3D_00s}}} 
    \subfloat[$t=0.3$\,s]
    {\makebox[0.48\hsize][c]{\includegraphics[trim = 0.3cm 0cm 0.5cm 0.7cm, clip=true, width=0.48\linewidth]{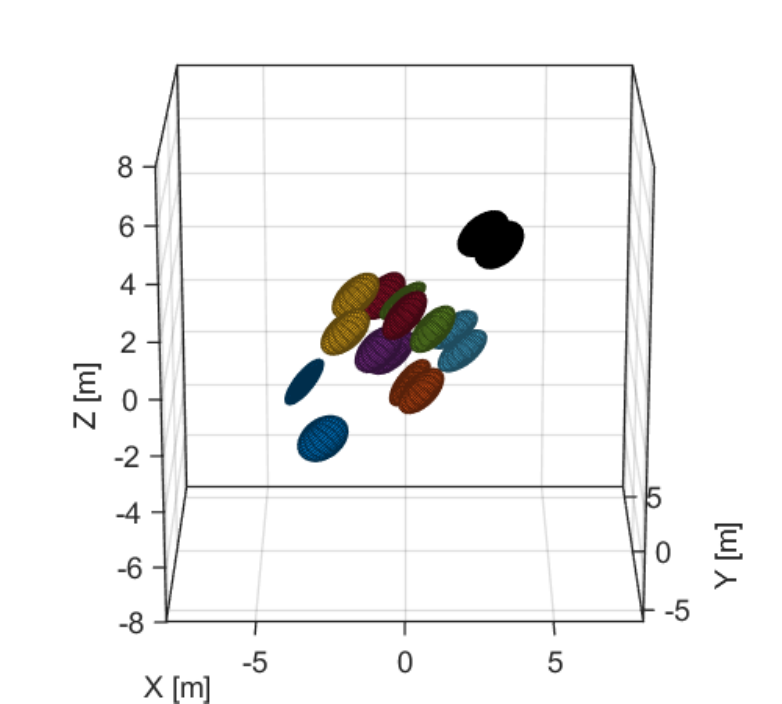}
    \label{subfig:sim_3D_03s}}}\\
    \subfloat[$t=0.5$\,s]
    {\makebox[0.48\hsize][c]{\includegraphics[trim = 0.3cm 0cm 0.5cm 0.7cm, clip=true, width=0.48\linewidth]{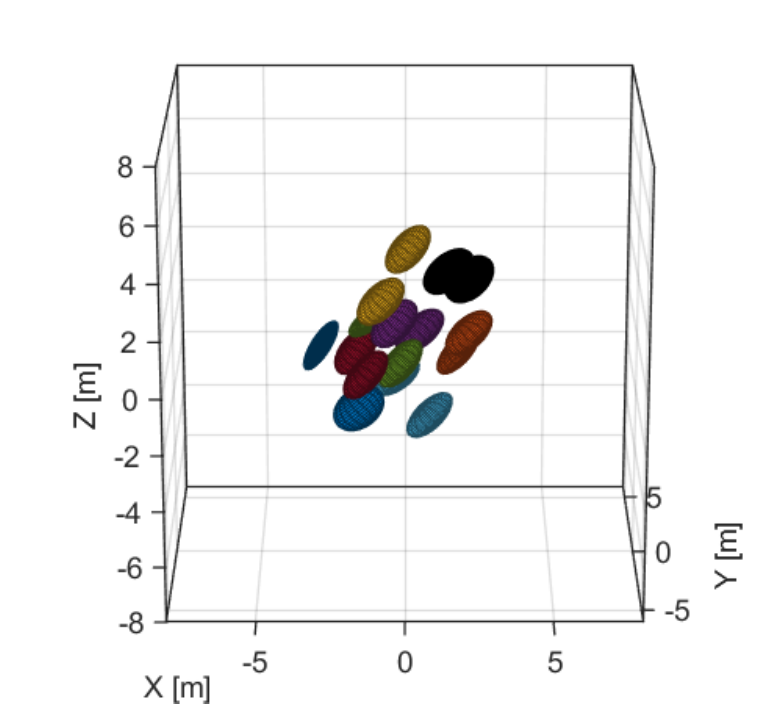}
    \label{subfig:sim_3D_05s}}} 
    \subfloat[$t=0.9$\,s]
    {\makebox[0.48\hsize][c]{\includegraphics[trim = 0.3cm 0cm 0.5cm 0.7cm, clip=true, width=0.48\linewidth]{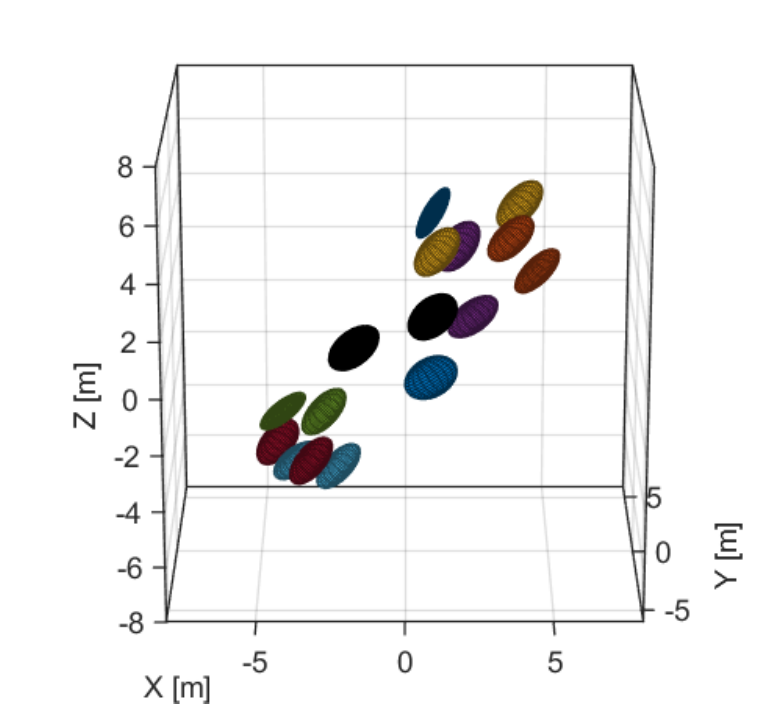}
    \label{subfig:sim_3D_09s}}} \\
    \subfloat[$t=1.2$\,s]
    {\makebox[0.48\hsize][c]{\includegraphics[trim = 0.3cm 0cm 0.5cm 0.7cm, clip=true, width=0.48\linewidth]{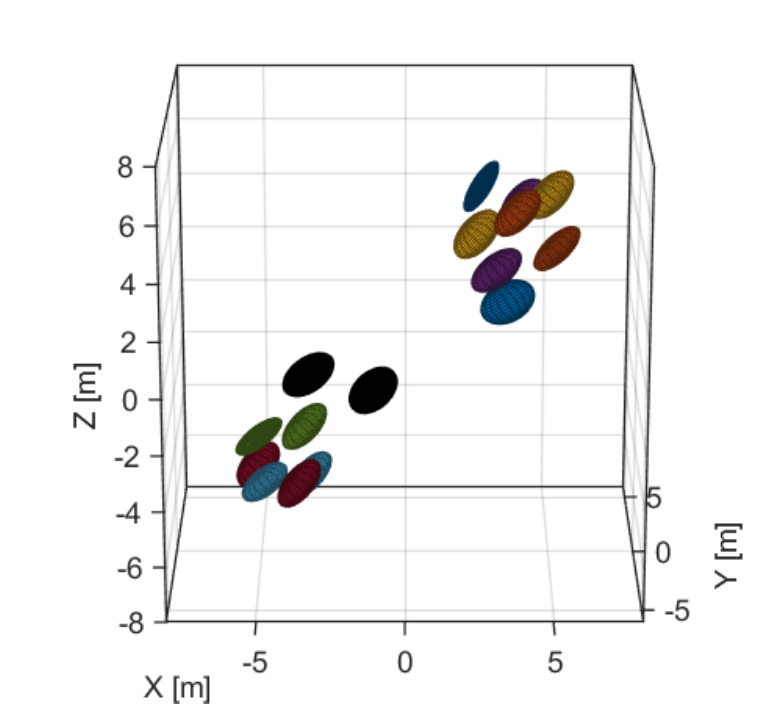}
    \label{subfig:sim_3D_12s}}} 
    \subfloat[$t=4.0$\,s]
    {\makebox[0.48\hsize][c]{\includegraphics[trim = 0.3cm 0cm 0.5cm 0.7cm, clip=true, width=0.48\linewidth]{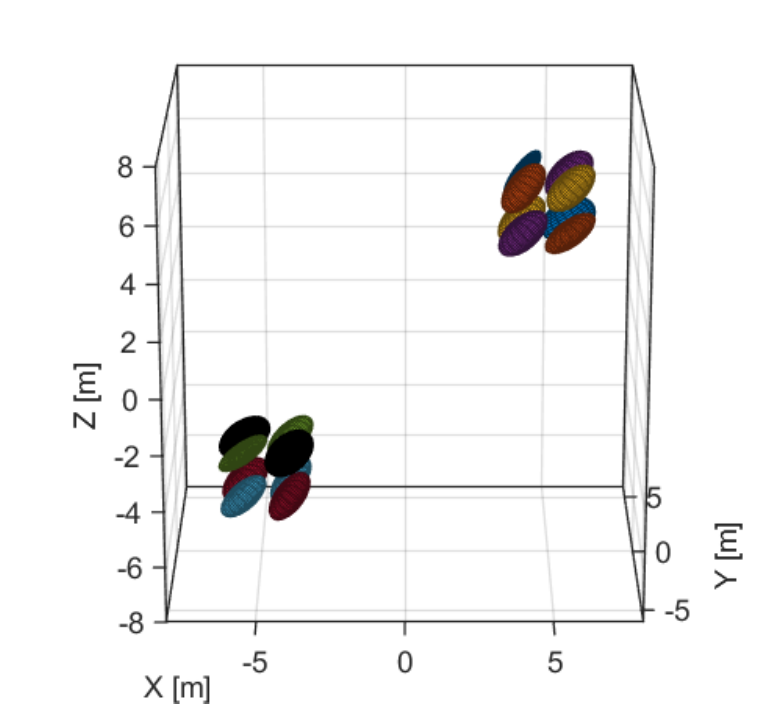}
    \label{subfig:sim_3D_40s}}}
    \caption{Snapshots of the simulation with the proposed collision avoidance method. The sixteen ellipsoidal rigid bodies are separated into two groups, which are located at the lower-left and the upper-right corner, respectively.} 
    \label{fig:sim_3D_snap}
\end{figure}

\begin{figure}
    \centering
    \includegraphics[trim = 0cm 0cm 0cm 0.5cm, clip=true, width=50mm]{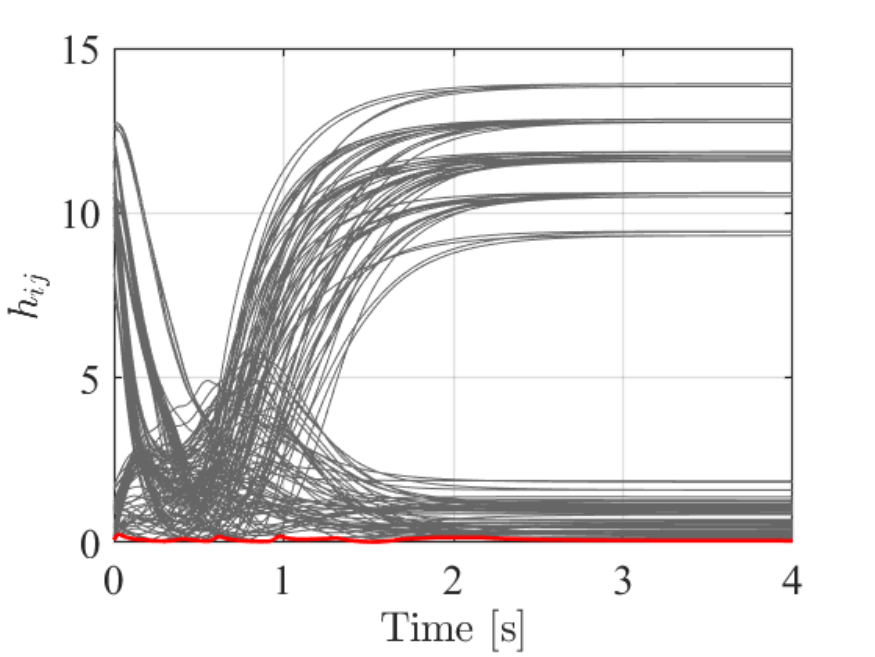}
    \caption{Evolution of CBFs $h_{ij}, \forall j \in \mc N \backslash \{ i\}, \forall i \in \mc N$. $h_{ij}$ taking the minimum value is highlighted as a red line, which keeps positive values.}
    \label{fig:sim_3D_CBF}
\end{figure}

Figure~\ref{fig:sim_3D_snap} shows the snapshots of the simulation result.
At $t=0.3$\,s and $t=0.5$\,s, two groups meet at the center of the field, where the proposed CBF rectifies the nominal input to prevent collision. 
During this time period, the CBFs, depicted in Fig.~\ref{fig:sim_3D_CBF}, take smaller values but remain positive, achieving collision avoidance.
After $t=0.9$\,s, each group approaches the final destination, where the position of each rigid body in the group is shuffled from the one in the initial coordination.
Still, all the rigid bodies smoothly converge to their final destinations while preventing collisions.

\section{Extension to the 2D Vehicle Model} \label{sec:vehicle}

\begin{figure}
    \centering
    \includegraphics[width = 40mm]{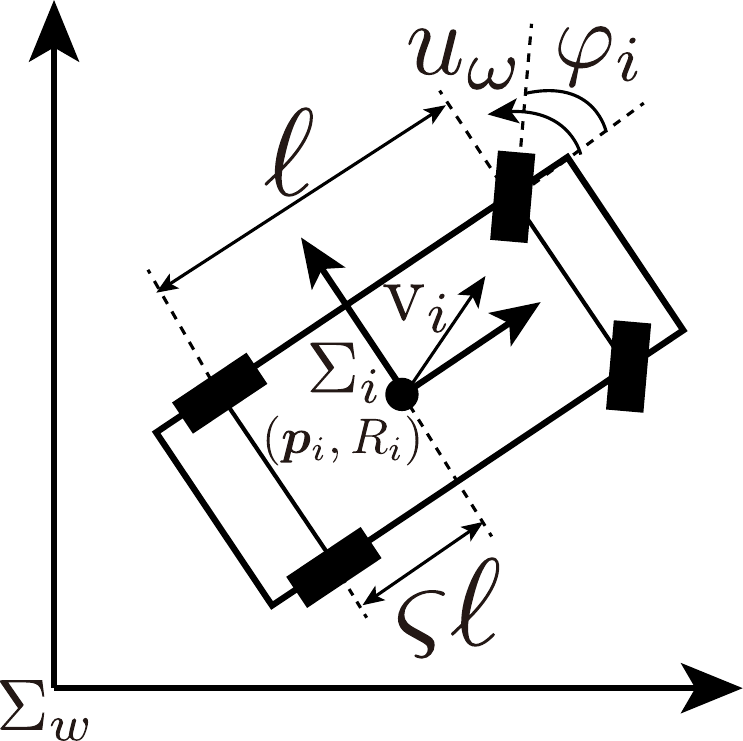}
    \caption{Vehicle model. The control input is the acceleration of the coordinate frame $\Sigma_i$, denoted as $\dot {\rm v}_i = u_{ai}$, and steering angular velocity $\dot{\varphi}_i = u_{\omega i}$.}
    \label{fig:vehicle}
\end{figure}

This section demonstrates that the proposed method achieves collision avoidance for a vehicle with nonholonomic dynamics by applying the methodology introduced in Section~\ref{sec:Preliminary} to the proposed CBF.
The considered vehicle model is described in Fig.~\ref{fig:vehicle}, where $(\bm{p}_i, R_i)$ is the relative pose of $\Sigma_i$ with respect to $\Sigma_w$.
$\varphi_i$ is the steering angle and ${\rm v}_i$ signifies the translational velocity of the center of mass. 
The length of the wheelbase is $\ell$, and the center of mass is located at a distance of $\varsigma \ell$ from the rear axle.
The control input is the acceleration, namely $\dot {\rm v}_i := u_{ai}$, and the steering angular velocity $\dot \varphi_i := u_{\omega i}$.
With assuming there is no wheel slip and applying the bicycle model as discussed in \cite[Chapter~2]{Rajamani2011}, the dynamics of the vehicles are expressed as 
\begin{subequations} \label{eq:vehicle_dynamics}
\begin{align}
\dot g_{i} &= g_{i} \hat V_{i}^b,~ \hat V_{i}^b = 
    \begin{bmatrix}
    \hat{\bm{\omega}}_i & \bm{v}_i \label{eq:vehicle_dynamics_RBM} \\
    0 & 0
    \end{bmatrix}, \\
\bm{v}_i &= \frac{{\rm v}_i}{\sqrt{1 + \left( \varsigma \tan{\varphi_i} \right)^2}} 
\begin{bmatrix}
    1 \\
    \varsigma \tan{\varphi_i}
\end{bmatrix}, \label{eq:vehicle_vel} \\
\bm{\omega}_i &= \frac{{\rm v}_i \tan{\varphi_i}}{\ell \sqrt{1 + \left( \varsigma \tan{\varphi_i} \right)^2}}, \label{eq:vehicle_omega} \\
\dot {\rm v}_i &= u_{ai}, \\
\dot \varphi_i &= u_{\omega i},
\end{align}
\end{subequations}
where we expressed the dynamics in \cite{Rajamani2011} by employing \eqref{eq:vehicle_dynamics_RBM} to suit this paper.

Differently from the dynamics \eqref{eq:dyn_RBM} considered in the previous section, the vehicle dynamics \eqref{eq:vehicle_dynamics} have relative degree two, i.e., the input $[u_{ai}~u_{\omega i}]^\top$ shows up in $\ddot h_{ij}$ but not in $\dot h_{ij}$.
To match the relative degree of the supporting hyperplane with the vehicle dynamics, we modify the dynamics of the supporting hyperplane as 
\begin{align} \label{eq:dyn_supporting_line_deg2}
\begin{bmatrix}
    \dot{\bm{z}}_{ij} \\ \dot{\bm{r}}_{ij}
\end{bmatrix} = 
\begin{bmatrix}
    0 & (I_d - \bm{z}_{ij} \bm{z}_{ij}^\top) \\
    0 & 0
\end{bmatrix}
\begin{bmatrix}
    \bm{z}_{ij} \\ \bm{r}_{ij}
\end{bmatrix}
+ \begin{bmatrix}
0 \\ I_d
\end{bmatrix}\bm{u}_{\bm{r}_{ij}}.
\end{align}
The nominal input for the supporting hyperplane is designed~as
\begin{align} \label{eq:input_supp_deg2}
    \bm{u}_{\mathrm{nom},\bm{r}_{ij}} = k_{\bm z} \frac{d}{dt} \left( \frac{\partial h_{ij}}{\partial \bm{z}_{ij}} \right),~k_{\bm z} > 0
\end{align}
with the intention that the input \eqref{eq:input_supp_deg2} integrated through the dynamics \eqref{eq:dyn_supporting_line_deg2} resembles the one in \eqref{eq:grad_ascent}, hence maximizing the distance between the supporting hyperplane and the other vehicle.
Similar to $\bm{z}_{i}=[\bm{z}_{i\hspace{0.4mm}i+1}^\top~\bm{z}_{i\hspace{0.4mm}i+2}^\top \cdots \bm{z}_{i\hspace{0.4mm}n}^\top]^\top$ introduced in the last paragraph of Section~\ref{ssec:Propose_method}, we introduce the augmented vector 
$\bm{r}_{i}=[\bm{r}_{i\hspace{0.4mm}i+1}^\top \cdots \bm{r}_{i\hspace{0.4mm}n}^\top]^\top$ and $\bm{u}_{\bm{r}_i} = [\bm{u}_{\bm{r}_{i\hspace{0.4mm}i+1}}^\top \cdots \bm{u}_{\bm{r}_{i\hspace{0.4mm}n}}^\top]^\top$.
Hereafter, we consider $X_i = (\bm{p}_i, R_i, {\rm v}_i, \varphi_i, \bm{z}_{i}, \bm{r}_{i})$ and $\bm{u}_i = [\bm{v}_{i}^\top~\bm{\omega}_i^\top~\bm{u}_{\bm{r}_{i}}^\top ]^\top$ as the state and input of rigid body~$i$.


The change of the dynamics of the supporting hyperplane replaces $\bm{u}_{\bm{z}_{ij}}$ in $\dot h_{ij}$, shown in \eqref{eq:dot_h_2D} in the case of $d=2$, into $\bm{r}_{ij}$ as follows, while other parts remain in the same value.
\begin{align} \label{eq:dot_h_2D_deg_two}
\begin{split}
    \dot{h}_{ij} &= \tilde{\bm{\zeta}}_{ij} \bm{\omega}_i + \bm{\eta}_{ij} R_i \bm{v}_i 
    \!+\! \bm{\mu}_{ij} \left(I_d-\bm{z}_{ij} \bm{z}_{ij}^\top \right) \bm{r}_{ij} \\
    &+ \tilde{\bm{\nu}}_{ij} \bm{\omega}_j + \bm{\xi}_{ij} R_j \bm{v}_j
\end{split}
\end{align}
Because the equation \eqref{eq:dot_h_2D_deg_two} contains neither the input $\bm{u}_i$ nor $\bm{u}_j$, $\dot{h}_{ij}$ cannot be served as a constraint for the input. Hence, following the concept in Section~\ref{sec:Preliminary}, we employ the following CBF for the considered vehicle model.
\begin{align} \label{eq:CBF_veh_ECBF}
    h_{ij,veh} = \dot h_{ij} + h_{ij}
\end{align}
Then, the constraint to achieve collision avoidance between rigid bodies $i$ and $j$ can be expressed as follows.
\begin{align} \label{eq:CBF_diff_veh_ECBF}
    \dot h_{ij,veh} = \ddot h_{ij} + \dot h_{ij} \geq -h_{ij,veh}
\end{align}

Let us take a closer look of \eqref{eq:CBF_diff_veh_ECBF} to confirm the condition \eqref{eq:CBF_diff_veh_ECBF} can be evaluated in a distributed manner as in \eqref{eq:QP_dist_i} and \eqref{eq:QP_dist_j}.
While $\dot h_{ij}$ is given in \eqref{eq:dot_h_2D_deg_two}, $\ddot h_{ij}$ is calculated as 
\begin{align}
\begin{split} \label{eq:ddot_h}
    \ddot h_{ij} &= \tilde{\bm{\zeta}}_{ij} \dot{\bm{\omega}}_i + \bm{\eta}_{ij}R_i \dot{\bm{v}}_i + \bm{\mu}_{ij} \left(I_d-\bm{z}_{ij} \bm{z}_{ij}^\top \right) \dot{\bm{r}}_{ij} \\ 
    &+ \tilde{\bm{\nu}}_{ij} \dot{\bm{\omega}}_j +  \bm{\xi}_{ij} R_j \dot{\bm{v}}_j + A (X_i, X_j),
\end{split} \\
\begin{split}
    A(X_i, X_j) &= \dot{\tilde{\bm{\zeta}}}_{ij} \bm{\omega}_i \!+\! \dot{\bm{\eta}}_{ij}R_i \bm{v}_i \!+\! \bm{\eta}_{ij} \dot{R}_i \bm{v}_i   \\
    &+\! \dot{\bm{\mu}}_{ij} \left(I_d\!-\!\bm{z}_{ij} \bm{z}_{ij}^\top \right) \bm{r}_{ij} \!+\! \bm{\mu}_{ij} \left(-2\bm{z}_{ij} \bm{z}_{ij}^\top \right) \bm{r}_{ij} \\
    &+ \dot{\tilde{\bm{\nu}}}_{ij} \bm{\omega}_j +  \dot{\bm{\xi}}_{ij} R_j \bm{v}_j + \bm{\xi}_{ij} \dot{R}_j \bm{v}_j,
\end{split}
\end{align}
where the term $A(X_i, X_j)$ only depends on the states of rigid bodies $i$ and $j$ and is independent of their control inputs.
Then, let us divide \eqref{eq:CBF_diff_veh_ECBF} into the following two conditions, where we substitute \eqref{eq:ddot_h} into \eqref{eq:CBF_diff_veh_ECBF}.
\begin{subequations} \label{eq:CBF_diff_veh_sep}
\begin{align}
\begin{split} \label{eq:CBF_diff_veh_sep_i}
    &\tilde{\bm{\zeta}}_{ij} \dot{\bm{\omega}}_i + \bm{\eta}_{ij}R_i \dot{\bm{v}}_i + \bm{\mu}_{ij} \left(I_d-\bm{z}_{ij} \bm{z}_{ij}^\top \right) \dot{\bm{r}}_{ij} \\ 
    &\hspace{8mm}\geq - \frac{1}{2}\left(A(X_i, X_j) + \dot h_{ij} + h_{ij,veh} \right)
\end{split}\\
    &\tilde{\bm{\nu}}_{ij} \dot{\bm{\omega}}_j \!+\!  \bm{\xi}_{ij} R_j \dot{\bm{v}}_j  \!\geq\! - \frac{1}{2}\left(A(X_i, X_j) \!+\! \dot h_{ij} \!+\! h_{ij,veh} \right)\label{eq:CBF_diff_veh_sep_j}
\end{align}
\end{subequations} 
Because $(\dot{\bm{\omega}}_i, \dot{\bm{v}}_i)$ and $(\dot{\bm{\omega}}_j, \dot{\bm{v}}_j)$ can be derived by differentiating \eqref{eq:vehicle_vel} and \eqref{eq:vehicle_omega} with time, the control inputs $\bm{u}_i$ and $\bm{u}_j$ appear in the left-hand side of \eqref{eq:CBF_diff_veh_sep_i} and \eqref{eq:CBF_diff_veh_sep_j}, respectively. 
Because each rigid body needs to calculate $A(X_i,X_j)$ and $\dot h_{ij}$, we need a similar requirement for communicated information to Assumption~\ref{ass:commu} to achieve distributed computation.
\begin{assumption} \label{ass:commu_deg2}
    Rigid body~$i\in \mc N$ can acquire $(\bm{p}_j, R_j, {\rm v}_j, \varphi_j)$ and $Q_j, \forall j\in \mc N\backslash \{i\}$. In addition, rigid body~$j$ can receive $[\bm{z}_{ij}^\top~ \bm{r}_{ij}^\top]^\top$ from rigid bodies $i,~\forall i\in \{1 \cdots j-1\}$.
\end{assumption}
Assumption~\ref{ass:commu_deg2} requires each vehicle to communicate the state describing its dynamics~\eqref{eq:vehicle_dynamics}. This is a similar requirement to Assumption~\ref{ass:commu}, which necessitates the communication of $(\bm{p}_i, R_i)$, the state of RBM \eqref{eq:dyn_RBM}. 


With Assumption~\ref{ass:commu_deg2}, the QP evaluated by rigid body $i$ is represented as 
\begin{subequations} 
\begin{align}
    &\bm{u}_{i}^* = \argmin_{\bm{u}_{i}}~\left\|\bm{u}_{i}-\bm{u}_{\mathrm{nom},i}\right\|_W^2 \\
    \begin{split}
    &\mbox{s.t.}~\tilde{\bm{\zeta}}_{ij} \dot{\bm{\omega}}_i + \bm{\eta}_{ij}R_i \dot{\bm{v}}_i + \bm{\mu}_{ij} \left(I_d-\bm{z}_{ij} \bm{z}_{ij}^\top \right) \dot{\bm{r}}_{ij} \\
    &\hspace{8mm}\geq - \frac{1}{2}\left(A(X_i, X_j) + \dot h_{ij} + h_{ij,veh} \right),~ \forall j>i,
    \end{split} \\
    \begin{split}
    &\hspace{5mm} \tilde{\bm{\nu}}_{ij} \dot{\bm{\omega}}_j +  \bm{\xi}_{ij} R_j \dot{\bm{v}}_j  \\
    &\hspace{8mm}\geq - \frac{1}{2}\left(A(X_i, X_j) + \dot h_{ij} + h_{ij,veh} \right),~ \forall j<i,
    \end{split}
\end{align}
\end{subequations}
with $W = \diag(\beta_{\rm v}, \beta_{\varphi}, I_{2d(n-i)})$.

\begin{figure}[t!]
    \vspace{-0.2cm}
    \centering
    \subfloat[$t=0$\,s]
    {\makebox[0.48\hsize][c]{\includegraphics[trim = 1.4cm 0cm 2.6cm 0.5cm, clip=true, width=0.48\linewidth]{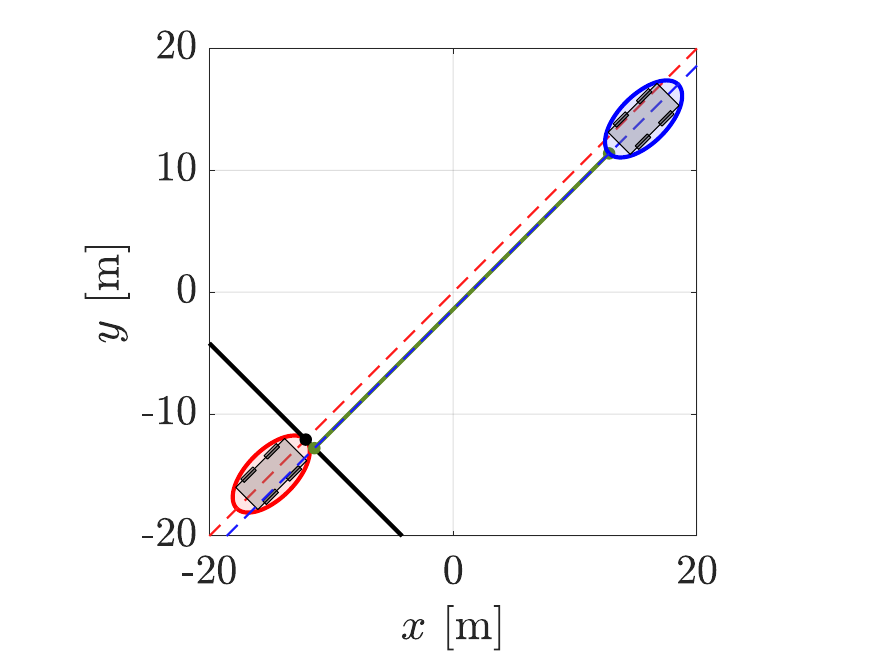}
    \label{subfig:sim_vehicle_0s}}} 
    \subfloat[$t=2$\,s]
    {\makebox[0.48\hsize][c]{\includegraphics[trim = 1.4cm 0cm 2.6cm 0.5cm, clip=true, width=0.48\linewidth]{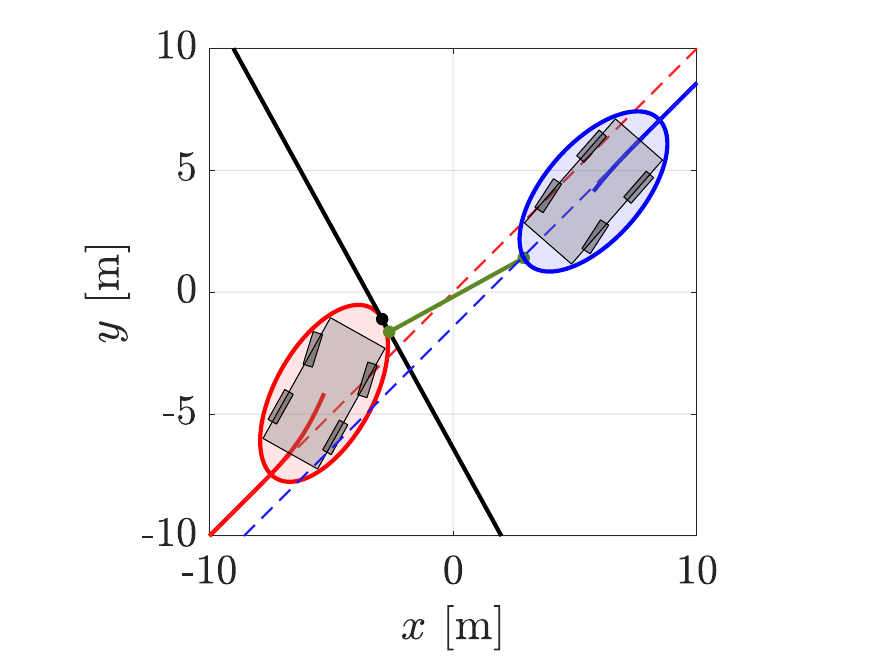}
    \label{subfig:sim_vehicle_2s}}}\\
    \subfloat[$t=4$\,s]
    {\makebox[0.48\hsize][c]{\includegraphics[trim = 1.4cm 0cm 2.6cm 0.5cm, clip=true, width=0.48\linewidth]{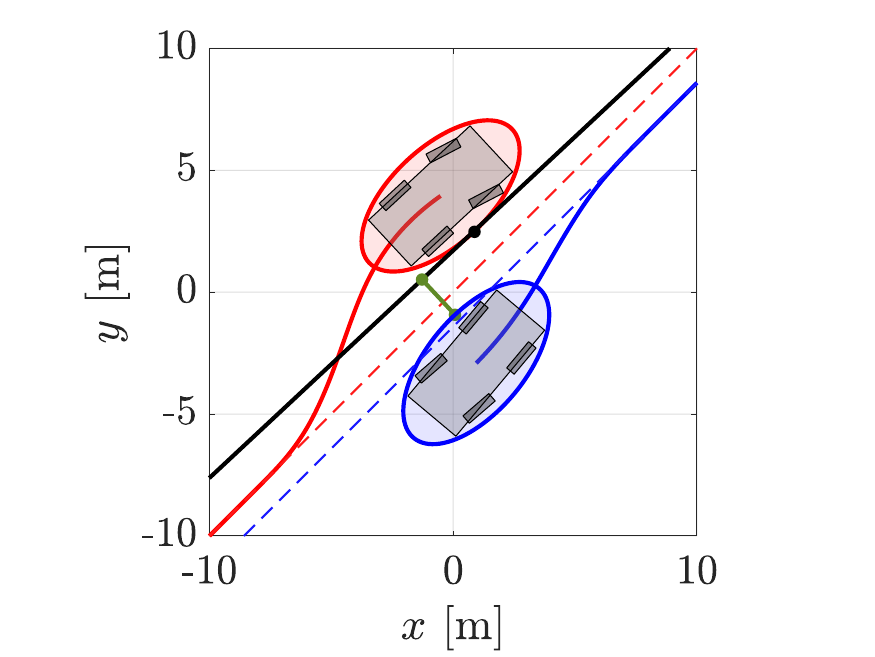}
    \label{subfig:sim_vehicle_4s}}} 
    \subfloat[$t=8$\,s]
    {\makebox[0.48\hsize][c]{\includegraphics[trim = 1.4cm 0cm 2.6cm 0.5cm, clip=true, width=0.48\linewidth]{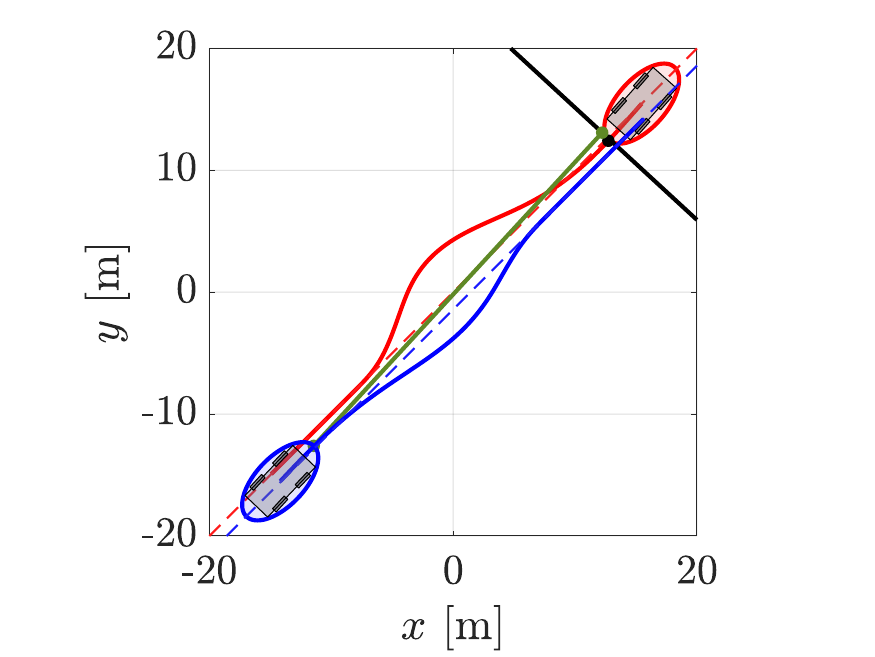}
    \label{subfig:sim_vehicle_8s}}} 
    \caption{Snapshots of the simulation with the two vehicles having the nonholonomic dynamics \eqref{eq:vehicle_dynamics}. Each vehicle traces a dashed line depicted in the same color as itself. Note that the snapshots (b) and (c) are zoomed in to detail how the vehicles avoid a collision.}
    \label{fig:sim_vehicle_snap}
\end{figure}

\begin{figure}
    \centering
    \includegraphics[width=50mm]{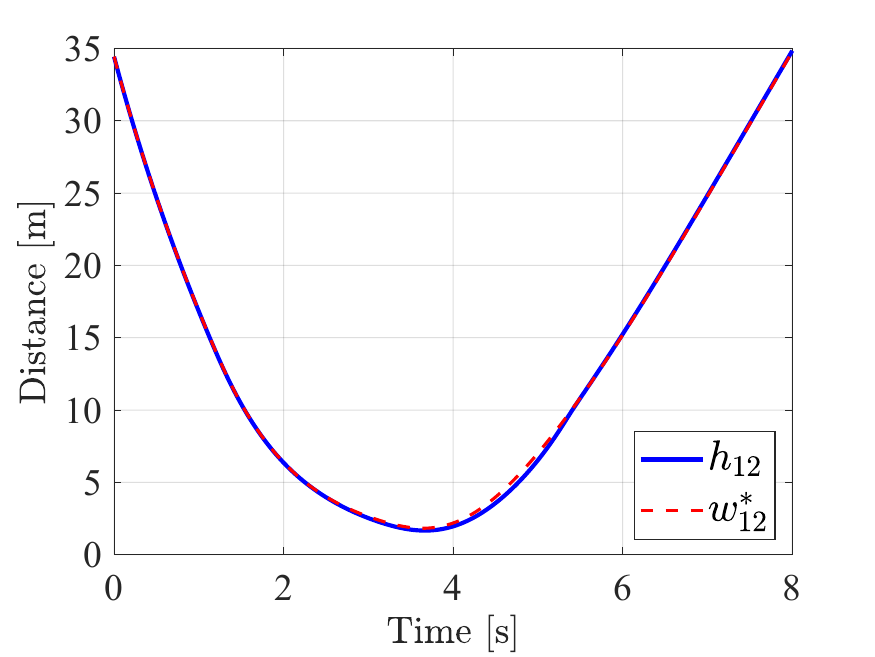}
    \caption{Evolution of CBF $h_{12}$ and the actual distance between vehicles $w_{12}^*$.}
    \label{fig:sim_vehicle_CBF}
\end{figure}

Having extended the proposed method to achieve collision avoidance for the rigid bodies with the vehicle dynamics \eqref{eq:vehicle_dynamics}, we are now ready to conduct a simulation to demonstrate the effectiveness of the method. 
Figure~\ref{fig:sim_vehicle_snap}\subref{subfig:sim_vehicle_0s} shows the initial configurations of the two vehicles, where each vehicle traces a dashed line depicted in the same color as itself.
The parameters of the vehicle shown in Fig.~\ref{fig:vehicle} are set as $\ell = 2.7\,{\rm m}$ and $\varsigma = 0.5$, respectively.
We model these vehicles as elliptical rigid bodies defined with $Q_1 = Q_2 = \diag(4, 2)$.
We set $\beta_{\rm v}=1$ and $\beta_{\varphi}=10$. 
The gain for the nominal input \eqref{eq:input_supp_deg2} for the supporting hyperplane is set to $k_{\bm z} = 1000/(1+h_{12}^2)$, which takes smaller value when two vehicles are far away. We limit the angular velocity of the unit vector $\bm{z}_{ij}$ in the range of $[-\frac{\pi}{3}, \frac{\pi}{3}]$ to avoid too fast rotational motion of the supporting hyperplane. 
The nominal control input for the vehicle is set as $u_{{\rm nom}, ai} = -({\rm v}_i - 5)$ and $u_{{\rm nom}, \omega i} = -0.1 e_1 - e_2 - 1.5 \varphi_i$. Here, $e_1$ is the lateral position error between the center of mass of the vehicle and the line to be traced.
$e_2=\frac{1}{2}\left( R_i^\top R_i^{\rm line} - (R_i^\top R_i^{\rm line})^\top \right)^\vee$ evaluates the error between the attitude of the vehicle and the direction of the line expressed by $R_i^{\rm line}$.

Figure~\ref{fig:sim_vehicle_snap} shows the snapshots of the simulation. 
As the two vehicles approach each other at $t=2\,{\rm s}$, the supporting hyperplane starts to rotate, and each vehicle turns the steering wheel to avoid a collision. At $t=4\,{\rm s}$, collision avoidance is successfully achieved, and vehicles return to the line to be traced at $t=8\,{\rm s}$.
Figure~\ref{fig:sim_vehicle_CBF} shows the value of $h_{12}$, the length of the green line in Fig.~\ref{fig:sim_vehicle_snap}, where its value keeps positive. Figure~\ref{fig:sim_vehicle_CBF} also depicts the actual distance between two vehicles, i.e. $w_{12}^*$. We can confirm that  $h_{12}$ tracks $w_{12}^*$ fast enough to eliminate the error between them.

%% file: text/6-Conclusion.tex
\section{Conclusion} \label{sec:conclusion}

This paper presented a collision avoidance method for ellipsoidal rigid bodies in $SE(2)$ and $SE(3)$, where we leverage the CBF employing the supporting hyperplane on the rigid body. 
We first formulated the problem with the rigid bodies governed by the rigid body motion. 
Then, we introduced a signed distance from the supporting hyperplane of an ellipsoid to the other ellipsoid to derive the condition that makes the ellipsoidal rigid bodies collision-free. 
However, we observed that the derived condition could render a smaller value than the actual distance between two rigid bodies if the supporting hyperplane is naively prepared. 
To prevent such a conservative evaluation, we designed the optimization problem that rotates the supporting hyperplane to maximize the signed distance from the supporting hyperplane to the other rigid body. We then proved that the maximum value of this optimization problem is equivalent to the actual distance between two ellipsoidal rigid bodies. This signed distance is leveraged as a CBF to achieve collision avoidance, where the supporting hyperplane is updated by a gradient-based input to eliminate conservativeness. The proposed CBF is implemented in a QP that allows each rigid body to compute the collision-free input in a distributed fashion under communication. The simulation studies demonstrated the validity of the proposed framework in both 2D and 3D environments. Finally, we showcased the presented CBF can be extended to the nonholonomic system.

%% file: text/Z_Appendix.tex
\section{Proof of Theorem~\ref{th:Prim_dual}} \label{Ap:Prim_dual}

This appendix proves Theorem~\ref{th:Prim_dual}. 
The dual function of the optimization problem~\eqref{eq:primal} is expressed as
\begin{subequations}
\begin{align}
\begin{aligned}
g(\lambda_i,\lambda_j,\bm{s}) = \inf_{\bm{x},\bm{y},\bm{w}}\left( \|\bm{w}\| \right.&\left.+\lambda_if_i(\bm{x}) + \lambda_jf_j(\bm{y}) \right.\\
&\left. +\bm{s}^\top(\bm{y}-\bm{x}-\bm{w})  \right)
\end{aligned}\\
=\begin{cases}
        \begin{aligned}
        \inf_{\bm{x}} &\left(\lambda_i f_i(\bm{x}) - \bm{s}^\top\bm{x}\right)\\&+ \inf_{\bm{y}} \left(\lambda_j f_j(\bm{y}) + \bm{s}^\top\bm{y}\right)
        \end{aligned} & \begin{aligned}
        \|\bm{s}\|&\leq 1,\\\lambda_i,\lambda_j &\geq 0
        \end{aligned}\\
    -\infty & \mbox{otherwise},
    \end{cases} \label{eq:Lagrangian}
\end{align}
\end{subequations}
with the Lagrange multipliers $\lambda_i$, $\lambda_j$, and $\bm{s}$.
%
Let us calculate the first term of \eqref{eq:Lagrangian}, i.e., $\inf_{\bm{x}} \left(\lambda_i f_i(\bm{x}) - \bm{s}^\top\bm{x}\right)$. Here, we introduce $\bar{\bm{x}}_i=\bar{Q}_i^{-1}\bm{x},~ \bar{\bm{p}}_i=\bar{Q}_i^{-1}\bm{p}_i$, and $\bar{\bm{s}}_i=\bar{Q}_i \bm{s}$ to simplify the equation as
\begin{align}
    &\inf_{\bm{x}} \left(\lambda_i f_i({\bm{x}}) - \bm{s}^\top{\bm{x}}\right)\nonumber\\
    =&\inf_{\bm{x}}\left( \lambda_i({\bm{x}}-\bm{p}_i)^\top\bar{Q}_i^{-2}({\bm{x}}-\bm{p}_i)-\lambda_i - \bm{s}^\top{\bm{x}}\right)\nonumber\\
    =&\inf_{\bm{x}}\left( \lam_i\left\|\bar{{\bm{x}}}_i \!-\! \left(\bar{\bm{p}}_i \!+\! \frac{1}{2\lambda_i}\bar{\bm{s}}_i\right)\right\|^2 \!-\! \bar{\bm{s}}_i^\top\bar{\bm{p}}_i \!-\! \frac{1}{4\lambda_i}\|\bar{\bm{s}}_i\|^2 \!-\! \lambda_i \right)\nonumber\\
    =&-\bar{\bm{s}}_i^\top\bar{\bm{p}}_i-\frac{\|\bar{\bm{s}}_i\|^2+4\lambda_i^2}{4\lambda_i}. \label{eq:dual_first}
\end{align}
Similarly, the second term in \eqref{eq:Lagrangian} can be transformed as
\begin{align} \label{eq:dual_second}
    \inf_{\bm{y}} \left(\lambda_j f_j(\bm{y}) + \bm{s}^\top\bm{y}\right)=\bar{\bm{s}}_j^\top\bar{\bm{p}}_j-\frac{\|\bar{\bm{s}}_j\|^2+4\lambda_j^2}{4\lambda_j},
\end{align}
with $\bar{\bm{p}}_j=\bar{Q}_j^{-1}\bm{p}_j$ and $\bar{\bm{s}}_j=\bar{Q}_j \bm{s}$.
From \eqref{eq:Lagrangian}, \eqref{eq:dual_first}, and \eqref{eq:dual_second}, the dual problem can be expressed as 
\begin{subequations}\label{eq:dual_1}
\begin{align}
    \max_{\bm{s},\lambda_i,\lambda_j}&-\bar{\bm{s}}_i^\top\bar{\bm{p}}_i -\frac{\|\bar{\bm{s}}_i\|^2+4\lambda_i^2}{4\lambda_i}
    +\bar{\bm{s}}_j^\top\bar{\bm{p}}_j-\frac{\|\bar{\bm{s}}_j\|^2+4\lambda_j^2}{4\lambda_j} \label{eq:dual_1_a}\\
    \mbox{s.t.}&~\bar{\bm{s}}_i=\bar{Q}_i\bm{s},~
    \bar{\bm{s}}_j=\bar{Q}_j \bm{s},~
    \|\bm{s}\|\leq 1,~
    \lambda_i,\lambda_j\geq 0.
\end{align}
\end{subequations}

We will next evaluate the second and fourth terms in \eqref{eq:dual_1_a}. For this goal, let us define a function $M(a, \tau)$ as
\begin{align}
    M(a,\tau)=-\frac{a+4\tau^2}{4\tau}~~~(a,\tau\geq 0),
\end{align}
and consider $\max_{a,\tau}M(a,\tau)$. 

In the case of $a>0$, the gradients of $M$ are 
\begin{align}
    \frac{\partial M}{\partial \tau} = \frac{(\sqrt{a}-2\tau)(\sqrt{a}+2\tau)}{4\tau^2},~\frac{\partial M}{\partial a} = - \frac{1}{4\tau}.
\end{align}
This implies that the function $M(a,\tau)$ has no extremum for all $a>0$, and for $\tau\geq0$ it has the maximum value at $\tau^*(a)={\sqrt{a}}/{2}$. 
Hence, the following equation holds. 
\begin{align}
    \max_{a,\tau}M(a,\tau)
    =\max_{a}-\sqrt{a}~~~(a>0)
\label{eq:a_neq_0}
\end{align}

When $a=0$ holds, the maximum value of $M(a,\tau)$ is 
\begin{align}
    \max_{a,\tau} M(a,\tau) = \max_\tau-\tau = 0.
    \label{eq:a_eq_0}
\end{align}
%
Given that \eqref{eq:a_eq_0} is equivalent to the result of substituting $a=0$ into \eqref{eq:a_neq_0}, we can unify them as 
\begin{align}
    \max_{a,\tau} M(a,\tau)=\max_{a}-\sqrt{a}~~~(a\geq 0).
    \label{eq:maxM_a_x}
\end{align}
Since the second and fourth terms in \eqref{eq:dual_1_a} are equal to $M(\|\bar{\bm{s}}_i\|^2,\lambda_i)$ and $M(\|\bar{\bm{s}}_j\|^2,\lambda_j)$, respectively, the problem \eqref{eq:dual_1} can be simplified by employing \eqref{eq:maxM_a_x} as 
\begin{subequations}\label{eq:dual_2}
\begin{align}
    \max_{\bm{s}}&-\bar{\bm{s}}_i^\top\bar{\bm{p}}_i-\left\|\bar{\bm{s}}_i\right\| + \bar{\bm{s}}_j^\top\bar{\bm{p}}_j-\left\|\bar{\bm{s}}_j\right\| \\
      \mbox{s.t.}&~\|\bm{s}\|\leq 1, \label{eq:dual_2_const2}\\
      &~\bar{\bm{s}}_i=\bar{Q}_i \bm{s},~\bar{\bm{s}}_j=\bar{Q}_j \bm{s}.
\end{align}
\end{subequations}

Let us parameterize $\bm{s}$ as 
\begin{align}
    \bm{s}&=\frac{\varrho}{\left\|\bar{Q}_i^{-1}\bm{z}_{ij}\right\|}\bar{Q}_i^{-1}\bm{z}_{ij},\label{eq:z_changed}
\end{align}
with $0\leq\varrho\leq 1$ to satisfy the constraint \eqref{eq:dual_2_const2}.
By substituting \eqref{eq:z_changed}, $\bar{\bm{s}} = \bar{Q}_i \bm{s}$ and $\bar{\bm{s}}_j = \bar{Q}_j \bm{s}$, the dual problem \eqref{eq:dual_2} can be expressed as 
\begin{subequations}\label{eq:dual_3}
\begin{align}
    \max_{\varrho,\bm{z}_{ij}}~&\varrho\underbrace{\frac{ -\left\|\bar{Q}_j \bar{Q}_i^{-1} \bm{z}_{ij}\right\| + (\bm{p}_j-\bm{p}_i)^\top\bar{Q}_i^{-1} \bm{z}_{ij} - 1}{\left\|\bar{Q}_i^{-1}\bm{z}_{ij}\right\|}}_{h_{ij}(\bm{z}_{ij})}\label{eq:dual_3_a}\\
    \mbox{s.t.}~& \|\bm{z}_{ij}\|= 1,
    \\&0\leq \varrho \leq 1.\label{eq:dual_3_c}
\end{align}
\end{subequations}
Notice that the objective function \eqref{eq:dual_3_a} can be expressed as $\varrho h_{ij}(\bm{z}_{ij})$. 
When two ellipsoids $\mathcal{E}_i$ and $\mathcal{E}_j$ have no overlap, there always exists $\bm{z}_{ij}$, namely a supporting hyperplane, realizing $h_{ij}(\bm{z}_{ij})>0$. 
Because the variable $\varrho$ is independent of $\bm{z}_{ij}$, $\varrho=1$ should hold to maximize \eqref{eq:dual_3_a}. 
Substituting $\varrho=1$ into \eqref{eq:dual_3_a} leads to the optimization problem~\eqref{eq:dual}.
This result concludes that the problem \eqref{eq:dual} is the dual of the optimization problem \eqref{eq:primal} if $\mc E_i \cap \mc E_j = \emptyset$.
Furthermore, because the optimization problem \eqref{eq:primal} satisfies the Slater's Condition \cite[Sec. 5.2.3]{convex_optimization}, the solution of \eqref{eq:primal} is equal to the solution of \eqref{eq:dual}, which completes the proof.

%% file: text/Z-Appendix_A.tex
\section{Proof of Lemma~\ref{lem:dot_h}} \label{Ap:dot_h_calc}

This appendix first derives the time derivatives of the proposed CBF in the case of $d=3$, then briefly discusses how $\dot{h}_{ij}$ differs in $d=2$. $h_{ij}(g_i,g_j, \bm{z}_{ij})$ is recalled here for ease of reference.
\begin{align} 
    &h_{ij}(g_i,g_j, \bm{z}_{ij}) =  \nonumber \\ 
    &-\frac{\left\|\bar{Q}_j \bar{Q}_i^{-1} \bm{z}_{ij}\right\|}{\left\|\bar{Q}_i^{-1}\bm{z}_{ij}\right\|} 
    + \frac{(\bm{p}_j-\bm{p}_i)^\top\bar{Q}_i^{-1} \bm{z}_{ij}}{\left\|\bar{Q}_i^{-1}\bm{z}_{ij}\right\|} 
    - \frac{1}{\left\|\bar{Q}_i^{-1}\bm{z}_{ij}\right\|} \label{eqA:CBF_Ap} 
\end{align}
Note that $\bar{Q}_i = R_i Q_i R_i^\top$ is a positive definite matrix since $Q_i$ is a diagonal matrix. 
Also, $\bar{Q}_i^\top \bar{Q}_i = \bar{Q}_i^2 = R_i Q_i^2 R_i^\top$ holds.

From \eqref{eq:dyn_RBM}, the dynamics of $\bm{p}_i$ and $R_i$ are expressed as
\begin{align}
    \dot{\bm{p}}_i = R_i \bm{v}_i, \label{eqA:dyn_pos} \\
    \dot R_i = R_i \hat{\bm{\omega}}_i. \label{eqA:dyn_rot}  
\end{align}
$\bm{z}_{ij}$ specifying a supporting hyperplane is governed by the following dynamics.
\begin{align} \label{eqA:dyn_v} %
    \dot{\bm{z}}_{ij} = \left(I_d-\bm{z}_{ij} \bm{z}_{ij}^\top \right) \bm{u}_{z_{ij}} 
\end{align}
In the case of $d=3$, we use the following properties the operator $\wedge$ and the rotation matrix $R \in SO(3)$ satisfy for any skew symmetric matrix $\hat{\bm{\omega}}$ and $\bm{a},\bm{b} \in \R^3$.
\begin{align}
    \hat{\bm{a}} \bm{b} &= -\hat{\bm{b}} \bm{a} \label{eqA:hat_exchange}\\
    R \hat{\bm{\omega}} R^\top &= (R \bm{\omega})^\wedge  \label{eqA:R_wedge}
\end{align}

We first derive the time derivative of the third term in \eqref{eqA:CBF_Ap}.
\begin{align}
    &\frac{d}{dt}\left( \frac{1}{\left\|\bar{Q}_i^{-1}\bm{z}_{ij}\right\|} \right) = 
    \frac{d}{dt} \left( \frac{1}{\sqrt{ \left(R_i^\top \bm{z}_{ij} \right)^\top Q_i^{-2} \left(R_i^\top \bm{z}_{ij}\right) } } \right) \nonumber \\
    &= -\frac{1}{2\left\|\bar{Q}_i^{-1}\bm{z}_{ij}\right\|^3 } 
    \left(2 \left(R_i^\top \bm{z}_{ij} \right)^\top Q_i^{-2} \frac{d}{dt} \left(R_i^\top \bm{z}_{ij}\right) \right) \nonumber \\
    &= -\frac{1}{\left\|\bar{Q}_i^{-1}\bm{z}_{ij}\right\|^3 } 
    \left( \left(R_i^\top \bm{z}_{ij} \right)^\top Q_i^{-2}  \left(R_i \hat{\bm{\omega}}_i \right)^\top \bm{z}_{ij} \right. \nonumber \\
    &\hspace{5mm} +\left. \left(R_i^\top \bm{z}_{ij} \right)^\top Q_i^{-2} R_i^\top \left(I_d - \bm{z}_{ij} \bm{z}_{ij}^\top \right) \bm{u}_{z_{ij}} \right) \label{eqA:third_1}
\end{align}
The equations \eqref{eqA:dyn_pos} and \eqref{eqA:dyn_rot} were substituted into \eqref{eqA:third_1}.
The first term of \eqref{eqA:third_1} can be expressed as
\begin{subequations} 
\begin{align}
    &\bm{z}_{ij}^\top  R_i Q_i^{-2} \left(R_i \hat{\bm{\omega}}_i \right)^\top \bm{z}_{ij} = 
    -\bm{z}_{ij}^\top R_i Q_i^{-2} \hat{\bm{\omega}}_i R_i^\top \bm{z}_{ij} \label{eqA:third_2_1} \\
    &= -\bm{z}_{ij}^\top R_i Q_i^{-2} R_i^\top R_i \hat{\bm{\omega}}_i R_i^\top \bm{z}_{ij} \label{eqA:third_2_2}\\
    &= -\bm{z}_{ij}^\top \bar{Q}_i^{-2} (R_i \bm{\omega}_i)^\wedge \bm{z}_{ij} \label{eqA:third_2_3}\\
    &= \bm{z}_{ij}^\top \bar{Q}_i^{-2} \hat{\bm{z}}_{ij} R_i \bm{\omega}_i. \label{eqA:third_2_4}
\end{align}
\end{subequations} 
Note that $\hat{\bm{\omega}}_i = -\hat{\bm{\omega}}_i^\top$ is utilized in \eqref{eqA:third_2_1}. 
In \eqref{eqA:third_2_2}, $R_i^\top R_i = I_d$ is substituted. In \eqref{eqA:third_2_3} and \eqref{eqA:third_2_4}, the properties \eqref{eqA:R_wedge} and \eqref{eqA:hat_exchange} are utilized, respectively. By substituting \eqref{eqA:third_2_4} into \eqref{eqA:third_1}, we obtain
\begin{align}
    &\frac{d}{dt}\left( \frac{1}{\left\|\bar{Q}_i^{-1}\bm{z}_{ij}\right\|} \right) \nonumber \\ 
    &= -\frac{\bm{z}_{ij}^\top \bar{Q}_i^{-2} \hat{\bm{z}}_{ij} R_i \bm{\omega}_i +  \bm{z}_{ij}^\top \bar{Q}_{i}^{-2} \left(I_d - \bm{z}_{ij} \bm{z}_{ij}^\top \right) \bm{u}_{z_{ij}}}{\left\|\bar{Q}_i^{-1}\bm{z}_{ij}\right\|^3 }.   \label{eqA:third_fin}
\end{align}

We next derive the time derivative of the second term in \eqref{eqA:CBF_Ap}.
\begin{subequations} \label{eqA:second_1}
\begin{align} 
    &\frac{d}{dt} \left( \frac{(\bm{p}_j-\bm{p}_i)^\top\bar{Q}_i^{-1} \bm{z}_{ij}}{\left\|\bar{Q}_i^{-1}\bm{z}_{ij}\right\|} \right) \nonumber \\
    &= \frac{d}{dt}\left( \frac{1}{\left\|\bar{Q}_i^{-1}\bm{z}_{ij}\right\|} \right) (\bm{p}_j-\bm{p}_i)^\top\bar{Q}_i^{-1} \bm{z}_{ij} \label{eqA:second_1_1} \\
    &+\frac{1}{\left\|\bar{Q}_i^{-1}\bm{z}_{ij}\right\|} \left( \frac{d}{dt}(\bm{p}_j-\bm{p}_i)^\top \right) \bar{Q}_i^{-1} \bm{z}_{ij} \label{eqA:second_1_2} \\
    &+ \frac{1}{\left\|\bar{Q}_i^{-1}\bm{z}_{ij}\right\|} (\bm{p}_j-\bm{p}_i)^\top \frac{d}{dt} \left( R_i Q_i^{-1} R_i^\top \bm{z}_{ij} \right) \label{eqA:second_1_3}
\end{align}
\end{subequations}
The derivative term of \eqref{eqA:second_1_1} has been obtained in \eqref{eqA:third_fin}.
By substituting \eqref{eqA:dyn_pos}, the equation \eqref{eqA:second_1_2} can be expressed as 
\begin{align} \label{eqA:second_2}
    &\frac{1}{\left\|\bar{Q}_i^{-1}\bm{z}_{ij}\right\|} \left( R_j\bm{v}_j - R_i\bm{v}_i \right)^\top \bar{Q}_i^{-1} \bm{z}_{ij} \nonumber \\
    &=\frac{1}{\left\|\bar{Q}_i^{-1}\bm{z}_{ij}\right\|} \bm{z}_{ij}^\top \bar{Q}_i^{-1}  \left( R_j\bm{v}_j - R_i\bm{v}_i \right).
\end{align}
The third term \eqref{eqA:second_1_3} can be calculated as
\begin{subequations}
\begin{align}
    &\frac{1}{\left\|\bar{Q}_i^{-1}\bm{z}_{ij}\right\|} (\bm{p}_j-\bm{p}_i)^\top \frac{d}{dt} \left( R_i Q_i^{-1} R_i^\top \bm{z}_{ij} \right) \label{eqA:second_3_1} \\
    \begin{split}\label{eqA:second_3_2}
    &= \frac{1}{\left\|\bar{Q}_i^{-1}\bm{z}_{ij}\right\|} (\bm{p}_j-\bm{p}_i)^\top 
    \Big( \left( R_i \hat{\bm{\omega}}_i\right) Q_i^{-1} R_i^\top \bm{z}_{ij}\\
    &+\! R_i Q_i^{-1}\! \left( R_i \hat{\bm{\omega}}_i\right)^\top \! \bm{z}_{ij} \!+\! R_i Q_i^{-1} R_i^\top (I_d\!-\!\bm{z}_{ij} \bm{z}_{ij}^\top ) \bm{u}_{z_{ij}} \Big).
    \end{split}
\end{align}
\end{subequations}
The first and the second term in \eqref{eqA:second_3_2} can be expressed as follows, where we omit $\| \bar{Q}_i^{-1}\bm{z}_{ij} \|^{-1}$.
\begin{subequations} \label{eqA:second_4}
\begin{align} 
    &(\bm{p}_j\!-\!\bm{p}_i)^\top R_i \hat{\bm{\omega}}_i  Q_i^{-1} R_i^\top \bm{z}_{ij} \!-\! 
    (\bm{p}_j\!-\!\bm{p}_i)^\top R_i Q_i^{-1} \hat{\bm{\omega}}_i R_i^\top \bm{z}_{ij} \nonumber \\
    \begin{split} \label{eqA:second_4_1}
    &= \bm{z}_{ij}^\top R_i Q_i^{-1}  R_i^\top R_i \hat{\bm{\omega}}_i^\top R_i^\top (\bm{p}_j\!-\!\bm{p}_i) \\
    &\hspace{5mm}- (\bm{p}_j\!-\!\bm{p}_i)^\top R_i Q_i^{-1}  R_i^\top R_i \hat{\bm{\omega}}_i R_i^\top \bm{z}_{ij} 
    \end{split} \\
    \begin{split} \label{eqA:second_4_2}
    &= -\bm{z}_{ij}^\top \bar Q_i^{-1}  \left( R_i \bm{\omega}_i \right)^\wedge (\bm{p}_j\!-\!\bm{p}_i) \\
    &\hspace{5mm}- (\bm{p}_j\!-\!\bm{p}_i)^\top \bar Q_i^{-1}  \left( R_i \bm{\omega}_i \right)^\wedge \bm{z}_{ij} 
    \end{split}\\
    &= \left(\bm{z}_{ij}^\top \bar Q_i^{-1} (\bm{p}_j\!-\!\bm{p}_i)^\wedge
    +(\bm{p}_j\!-\!\bm{p}_i)^\top \bar Q_i^{-1}  \hat{\bm{z}}_{ij} \right) R_i \bm{\omega}_i \label{eqA:second_4_3}
\end{align}
\end{subequations}
Note that we transposed the first term and substituted $R_i^\top R_i$ into both terms in \eqref{eqA:second_4_1}.
In \eqref{eqA:second_4_2}, $\bar Q_i^{-1} = R_i Q_i^{-1} R_i^\top$ and \eqref{eqA:R_wedge} are used. Lastly, the property \eqref{eqA:hat_exchange} is applied in \eqref{eqA:second_4_3}.
By substituting \eqref{eqA:second_2}, \eqref{eqA:second_3_2} and \eqref{eqA:second_4_3} into \eqref{eqA:second_1}, we obtain
\begin{align} 
    &\frac{d}{dt} \left( \frac{(\bm{p}_j-\bm{p}_i)^\top\bar{Q}_i^{-1} \bm{z}_{ij}}{\left\|\bar{Q}_i^{-1}\bm{z}_{ij}\right\|} \right) \nonumber \\
    &= \frac{d}{dt}\left( \frac{1}{\left\|\bar{Q}_i^{-1}\bm{z}_{ij}\right\|} \right) (\bm{p}_j-\bm{p}_i)^\top\bar{Q}_i^{-1} \bm{z}_{ij} \nonumber  \\
    &+\frac{1}{\left\|\bar{Q}_i^{-1}\bm{z}_{ij}\right\|} \bm{z}_{ij}^\top \bar{Q}_i^{-1}  \left( R_j\bm{v}_j - R_i\bm{v}_i \right) \nonumber \\
    &+ \frac{1}{\left\|\bar{Q}_i^{-1}\bm{z}_{ij}\right\|} 
    \Big( (\bm{p}_j\!-\!\bm{p}_i)^\top \bar Q_i^{-1} \left(I_d-\bm{z}_{ij} \bm{z}_{ij}^\top\right) \bm{u}_{z_{ij}} \nonumber \\
    &+ \left(\bm{z}_{ij}^\top \bar Q_i^{-1} (\bm{p}_j\!-\!\bm{p}_i)^\wedge
    +(\bm{p}_j\!-\!\bm{p}_i)^\top \bar Q_i^{-1}  \hat{\bm{z}}_{ij} \right) R_i \bm{\omega}_i \Big) \label{eqA:second_fin},
\end{align}
where $d/dt (\|\bar{Q}_i^{-1}\bm{z}_{ij}\|^{-1} )$ is given from \eqref{eqA:third_fin}.

Finally, we derive the time derivative of the first term in \eqref{eqA:CBF_Ap}. 
Here, we introduce $\sigma:=\left\|\bar{Q}_i^{-1}\bm{z}_{ij}\right\| \left\|\bar{Q}_j \bar{Q}_i^{-1} \bm{z}_{ij} \right\|$ to make the notations simple.
\begin{subequations} \label{first_1}
\begin{align}
    &\frac{d}{dt}\! \left( \frac{ \left\|\bar{Q}_j \bar{Q}_i^{-1} \! \bm{z}_{ij}\right\|}{\left\|\bar{Q}_i^{-1}\bm{z}_{ij}\right\|} \right) 
    \!=\! \left\|\bar{Q}_j \bar{Q}_i^{-1} \! \bm{z}_{ij} \right\| \frac{d}{dt} \!\left( \frac{1}{\left\|\bar{Q}_i^{-1}\bm{z}_{ij}\right\|} \right) \\
    &+ \frac{1}{\sigma}
    \left( \bm{z}_{ij}^\top \bar{Q}_i^{-1} \bar{Q}_j^2\right) \left( \frac{d}{dt}\left( R_i Q_i^{-1}R_i^\top \bm{z}_{ij} \right) \right) \label{eqA:first_1_2} \\
    &+ \frac{1}{2\sigma}
    \left( \bm{z}_{ij}^\top \bar{Q}_i^{-1}\right)  \left( \frac{d}{dt}\left(R_j Q_j^2 R_j^\top \right)\right) \left(\bar{Q}_i^{-1}\bm{z}_{ij}\right) \label{eqA:first_1_3}
\end{align}
\end{subequations}
Because the derivative term of the second term \eqref{eqA:first_1_2} is the same as the one in \eqref{eqA:second_3_1}, the second term \eqref{eqA:first_1_2} can be calculated as follow, where we utilize the same techniques utilized in \eqref{eqA:second_3_2} and \eqref{eqA:second_4}.
\begin{subequations}
\begin{align}
    &\frac{1}{\sigma}
    \left( \bm{z}_{ij}^\top \bar{Q}_i^{-1} \bar{Q}_j^2\right) \left( \frac{d}{dt}\left( R_i Q_i^{-1}R_i^\top \bm{z}_{ij} \right) \right) \\
    \begin{split}
    &= \frac{1}{\sigma} \left( \bm{z}_{ij}^\top \bar{Q}_i^{-1} \bar{Q}_j^2\right) \Big( \left( R_i \hat{\bm{\omega}}_i\right) Q_i^{-1} R_i^\top \bm{z}_{ij} \\
    &+\! R_i Q_i^{-1} \left( R_i \hat{\bm{\omega}}_i\right)^\top \! \bm{z}_{ij} \!+\! R_i Q_i^{-1} R_i^\top (I_d\!-\!\bm{z}_{ij} \bm{z}_{ij}^\top) \bm{u}_{z_{ij}} \Big)
    \end{split} \\
    \begin{split} \label{eqA:first_2}
    &= \frac{1}{\sigma} \left( \bm{z}_{ij}^\top \bar{Q}_i^{-1} \bar{Q}_j^2\right) 
    \left( -\left(\bar{Q}_i^{-1} \hat{\bm{z}}_{ij} \right)^\wedge + \bar{Q}_i^{-1} \hat{\bm{z}}_{ij}  \right) R_i \bm{\omega}_i \\
    &+ \frac{1}{\sigma}\left( \bm{z}_{ij}^\top \bar{Q}_i^{-1} \bar{Q}_j^2\right) 
    \bar{Q}_i^{-1} \left(I_d-\bm{z}_{ij} \bm{z}_{ij}^\top\right) \bm{u}_{z_{ij}}
    \end{split}
\end{align}
\end{subequations}
The third term \eqref{eqA:first_1_3} can be calculated as 
\begin{subequations} 
\begin{align}
    &\frac{1}{\sigma}
    \left( \bm{z}_{ij}^\top \bar{Q}_i^{-1}\right) \left(R_j Q_j^2 \left(R_j \hat{\bm{\omega}}_j \right)^\top \right) \left(\bar{Q}_i^{-1}\bm{z}_{ij}\right)\\
    &=\frac{1}{\sigma}
    \left( \bm{z}_{ij}^\top \bar{Q}_i^{-1}\right) \left(R_j Q_j^2 R_j^\top R_j \hat{\bm{\omega}}_j^\top R_j^\top  \right) \left(\bar{Q}_i^{-1}\bm{z}_{ij}\right) \\
    &=-\frac{1}{\sigma}
    \left( \bm{z}_{ij}^\top \bar{Q}_i^{-1}\right) \bar{Q}_j^2 \left(R_j \bm{\omega}_j \right)^\wedge \left(\bar{Q}_i^{-1}\bm{z}_{ij}\right)\\
    &=\frac{1}{\sigma}
    \left( \bm{z}_{ij}^\top \bar{Q}_i^{-1}\right) \bar{Q}_j^2 \left(\bar{Q}_i^{-1}\bm{z}_{ij}\right)^\wedge  R_j \bm{\omega}_j. \label{eqA:first_3}
\end{align}
\end{subequations}
By substituting \eqref{eqA:first_2} and \eqref{eqA:first_3} into \eqref{first_1}, we obtain
\begin{subequations} \label{eqA:first_fin}
\begin{align}
    &\frac{d}{dt}\! \left( \frac{ \left\|\bar{Q}_j \bar{Q}_i^{-1} \! \bm{z}_{ij}\right\|}{\left\|\bar{Q}_i^{-1}\bm{z}_{ij}\right\|} \right) 
    \!=\! \left\|\bar{Q}_j \bar{Q}_i^{-1} \! \bm{z}_{ij} \right\| \frac{d}{dt} \!\left( \frac{1}{\left\|\bar{Q}_i^{-1}\bm{z}_{ij}\right\|} \right) \\
    &+ \frac{1}{\sigma} \left( \bm{z}_{ij}^\top \bar{Q}_i^{-1} \bar{Q}_j^2\right) 
    \left( -\left(\bar{Q}_i^{-1} \hat{\bm{z}}_{ij} \right)^\wedge + \bar{Q}_i^{-1} \hat{\bm{z}}_{ij}  \right) R_i \bm{\omega}_i \\
    &+ \frac{1}{\sigma}\left( \bm{z}_{ij}^\top \bar{Q}_i^{-1} \bar{Q}_j^2\right) 
    \bar{Q}_i^{-1} \left(I_d-\bm{z}_{ij} \bm{z}_{ij}^\top\right) \bm{u}_{z_{ij}} \\
    &+\frac{1}{\sigma}
    \left( \bm{z}_{ij}^\top \bar{Q}_i^{-1}\right) \bar{Q}_j^2 \left(\bar{Q}_i^{-1}\bm{z}_{ij}\right)^\wedge  R_j \bm{\omega}_j. 
\end{align}
\end{subequations}
The time derivative of the proposed CBF in $d=3$ can be obtained by combining \eqref{eqA:third_fin}, \eqref{eqA:second_fin}, and \eqref{eqA:first_fin} together. Note that $\dot h_{ij}$ presented in \eqref{eq:dot_h} and \eqref{eq:diff_results} divides the terms of $\dot h_{ij}$ so that the coefficients of each control input, namely $\bm{\omega}_i$, $\bm{v}_{i}$, $\bm{u}_{z_{ij}}$, $\bm{\omega}_j$, and $\bm{v}_{j}$, are easily understandable.


Finally, we briefly discuss the time derivative of $h_{ij}$ in the case of $d=2$. Because most of the equations are the same as those of $d=3$ shown above, we only explain what causes the difference between $d=3$ and $d=2$. 

The disagreement of $\dot{h}_{ij}$ between $d=2$ and $d=3$ mainly stems from the operator $\wedge$, which renders a different result, as shown in \eqref{eq:wedge}. As a consequence, with $d=2$, the properties of the operator $\wedge$ are distinct from that of $d=3$ in \eqref{eqA:hat_exchange} and \eqref{eqA:R_wedge}, as shown below with $a\in \R,~\bm{b} \in \R^2$, and $R\in SO(2)$.
\begin{align} 
\hat{a}\bm{b} &= \hat{1}\bm{b} a \label{eqA:hat_exchange_2D} \\
R\hat{\omega}R^\top &= \hat{\omega} \label{eqA:R_wedge_2D}
\end{align}
Hence, $\dot h_{ij}$ in $d=2$ takes a bit different form from \eqref{eq:dot_h}, as shown in \eqref{eq:dot_h_2D}. 
More concretely, the difference between \eqref{eqA:hat_exchange} and \eqref{eqA:hat_exchange_2D} alters $(\hat{\cdot})$ in \eqref{eq:diff_results_omegai} and \eqref{eq:diff_results_z} into $-\hat{1}(\cdot)$ in \eqref{eq:diff_results_omegai_2D} and \eqref{eq:diff_results_omegaj_2D}, where $(\cdot)$ denotes a vector with the length of $d$. 
In addition, because of the difference between \eqref{eqA:R_wedge} and \eqref{eqA:R_wedge_2D}, $R_i$ and $R_j$ preceding $\bm{\omega}_i$ and $\bm{\omega}_j$ in \eqref{eq:dot_h} disappear in \eqref{eq:dot_h_2D}. This completes the proof.